\documentclass[oneside,envcountsame,envcountsect]{llncs}

\usepackage{palatino}
\usepackage{amsmath}
\usepackage[dvips]{pstcol}
\usepackage{amssymb}
\usepackage{amscd}
\usepackage{graphics}
\usepackage{epsfig}
\usepackage{algorithm2e}
\usepackage{alltt}

\renewcommand{\marginpar}[1]{{}}

 \newcounter{myenumi}

\usepackage{fancyhdr}
\newcommand{\myversion}{}%
\fancypagestyle{plain}{%
\fancyhf{} %
\fancyfoot[L]{\myversion}
\fancyfoot[R]{\thepage} %
\renewcommand{\headrulewidth}{0pt} \renewcommand{\footrulewidth}{0pt}
\fancyhead[R]{}%
}

\fancypagestyle{empty}{%
\fancyhf{} %
\renewcommand{\headrulewidth}{0pt} \renewcommand{\footrulewidth}{0pt}
\fancyfoot[C]{\myversion} %
}

\usepackage{amsmath} %
\usepackage{amssymb}

\usepackage{xspace} %

\newcommand{\ie}{{\it i.e.}}
\newcommand{\bG}{\ensuremath{\mathbf G}\xspace}
\newcommand{\bH}{\ensuremath{\mathbf H}\xspace}
\newcommand{\bK}{\ensuremath{\mathbf K}\xspace}

\newcommand{\figdelta}{\delta}
\newcommand{\figgamma}{\gamma}

\newcommand{\cF}{\ensuremath{\mathcal{F}}}

\renewcommand{\epsilon}{\varepsilon}

\def\card{\mbox{card}} \def\cG{{\cal G}} 
  
\def\ra{\rightarrow}

\newcommand{\macro}[2]{ \providecommand{#1}{{\ensuremath{#2}}\xspace}}
\macro{\bF}{{\mathbf F}}
\macro{\grs}{\mathrel{{\mathcal R}}}
\macro{\gfam}{\cF}%
\macro{\N}{\mathbb N}

\newcommand{\graph}{\cG\xspace} %

\macro{\etiq}{L} %
\newcommand{\lgraph}{\ensuremath{\cG_\etiq}\xspace} %

\newcommand{\gmin}{{\ensuremath{\mathcal G_{\mathrm{min}}}}\xspace}

\macro{\simtaille}{\sim^{\mbox{\tiny\textsc{Taille}}}}

\newcommand{\RR}{\ensuremath{\mathrel{\mathcal R}}} 
\newcommand{\romanitems}{\renewcommand{\labelenumii}{(\roman{enumii})}}

\newcounter{theoenumcounter}
\def\theoenumlabel{\thetheorem.\roman{theoenumcounter}}

\newenvironment{theoenum}{\begin{list}{\theoenumlabel}{%
      \usecounter{theoenumcounter}      
}      %
  }{\end{list}}

\newcounter{rrulec}
\newenvironment{rrule}[1]{%
  \begin{list}{\therrulec~:}
    { 
      \usecounter{rrulec}
      \newcommand{\therrule}{#1}
    }
  }
  {\end{list}}
\newcommand{\ritem}[4][\therrule\arabic{rrulec}]{
  \renewcommand{\therrulec}{#1}
  \pagebreak[3]
  \item {\bf #2 }\nopagebreak[4] 
    \begin{itemize}
    \item[] \underline{\em Precondition~:}
      \begin{itemize}
        #3
      \end{itemize}
    \item[] \underline{\em Relabelling~:}
      \begin{itemize}
        #4
      \end{itemize}
    \end{itemize}{\vskip 0.5cm}
}

\renewenvironment{proof}{\par\noindent{\bf Proof.}}{\hfill%
    $\square$\medskip\endtrivlist} 

\title{
  Termination Detection\\ 
  of Local Computations
}
 \author{Emmanuel Godard \inst{1}, Yves M{\'e}tivier \inst{2} and 
Gerard Tel \inst{3}
}

 \institute{
\email{emmanuel.godard@lif.univ-mrs.fr}\\
Laboratoire d'Informatique Fondamentale \\
CNRS (UMR 6166) -- 
Universit{\'e} de Provence\\
39 rue Joliot-Curie\\ 
13453, Marseille Cedex 13, France\\
\and
\email{metivier@labri.fr}\\
LaBRI\\ 
Universit{\'e} Bordeaux 1,  ENSEIRB,\\
 351 cours de la Lib{\'e}ration\\ 
 33405 Talence, France
\and
\email{gerard@cs.uu.nl}\\
Department of Computer Science\\
University of Utrecht\\
P.O. Box 80.089, 3508 TB Utrecht\\
The Netherlands
}
\date{ } 

\begin{document}
\setcounter{tocdepth}{3} %
\maketitle

\begin{abstract}

Contrary to the sequential world, the processes involved in a
distributed system do not necessarily know when a computation is
{\em globally} finished. This paper investigates the problem of the
detection of the termination of local computations.

We define four types of termination detection: no 
detection, detection of the local termination, detection by a distributed
observer, detection of the global termination. We give a complete
characterisation (except in the local termination detection case where
a partial one is given) for each of this termination detection and
show that they define a strict hierarchy. These results emphasise the
difference between computability of a distributed task and termination
detection. 

Furthermore, these characterisations encompass all standard criteria
that are usually formulated :
topological restriction (tree, rings, or triangulated networks ...),
topological knowledge (size, diameter ...),
and local knowledge to distinguish nodes (identities, sense of direction).
These results are now presented as corollaries of  generalising 
theorems.
As a very special and important case, the techniques are also applied
to the election problem. Though given in the model of local
computations, these results can give qualitative insight for similar
results in other standard models.

 The necessary conditions involve graphs covering and quasi-covering;
 the sufficient conditions (constructive local computations)
 are based upon an enumeration algorithm of Mazurkiewicz and a stable
 properties detection algorithm of Szymanski, Shi and Prywes.

\end{abstract}
\tableofcontents
\par
\newpage
\section{Introduction}
This paper presents results concerning two fundamental problems in the
area of distributed computing:  the termination detection problem and the
election problem. 
The proofs are done in the model of local computations and use mainly
common results and tools.
Namely, they use Mazurkiewicz' algorithm \cite{MazurEnum}, the
Szymanski-Shi-Prywes 
algorithm \cite{SSP}, coverings and quasi-coverings of graphs. 
\noindent
\subsection{The Model}
We consider networks of processors with arbitrary
 topology. A network is represented as a connected, 
undirected graph where vertices
 denote processors and edges denote direct communication links.  Labels 
are attached to vertices and edges. 
The identities of the vertices, 
a distinguished vertex, the number of 
 processors, the diameter of the graph or the topology are examples of
 labels
attached to vertices; weights, marks for encoding a spanning tree or
 the sense of direction are examples of labels attached to edges.

The basic computation step is to  modify labels  \emph{locally}, that is, 
on a subgraph of fixed radius $1$ of the given graph, according to certain
rules depending on the subgraph only ({\em local computations}). The
relabelling is performed until no more transformation is possible, \ie, until
a normal form is obtained. This is a model first proposed by
A. Mazurkiewicz \cite{Mazur}.

This model  has numerous interests. As any rigorously
defined model, it gives an abstract tool 
to think about some problems in the field of distributed computing
independently of the wide variety of models used to represent
distributed systems \cite{lamport}. As classical models in
programming, it enables to build and to prove complex systems, and so,
to get them right.
And quoting D. Angluin in \cite{Angluin}, this kind of model makes it
possible to put forward  phenomena common to other models.
It is true that this model is strictly stronger
than other standard models (like message passing systems), but then,
impossibility results remains true in weaker models. Furthermore, any
positive solution in this model may guide 
the research of a solution in a weaker model or be implemented 
in a weaker model using randomised algorithms. Finally,  this model 
gives nice properties and examples using classical combinatorial 
material, hence
we believe this model has a very light overhead in order to 
understand and to explain distributed problems.

We acknowledge, and underline, that the results presented here might be
quantitatively different from other models, but we claim that they are not
significantly different: they are qualitatively similar, as are all the
impossibility results proved in different models since the seminal
work of Angluin. All of them use the same ``lifting technique'',
even though not on exactly the same kind of graph morphism
\cite{Angluin,MazurEnum,YKsolvable,BVselfstab}. Thus it seems possible to extend the
general results of this paper to more standard models like the
``message passing model''. Moreover, this direction has already given
some results \cite{CGMT,CMelection,CGMlocalterm}. Note also that all the
questions addressed in this paper are not specific of the model of
local computations. E.g, is there a unique (universal) algorithm that
can solve the election problem on the family \gmin of networks that admit an
election algorithm? Though this very\marginpar{ML ?} set \gmin can be
different depending on 
the model of computations that is used, we claim that the generic
answer is no and that our main impossibility result can be extended to
any other model. The reader should note that this question has not been
previously thoroughly answered in any model (see the discussion about
the election problem on Section~\ref{BVelection}).

\subsection{Related Works} 
Among models related to our model there are local
computation systems as defined by Rosenstiehl et al.~\cite{Rosen},
Angluin \cite{Angluin}, Yamashita and Kameda \cite{Yaka1}, Boldi
and Vigna \cite{BV,BV0} and Naor and Stockmeyer \cite{Stockmeyer}.  In
\cite{Rosen} a synchronous model is considered, where vertices
represent (identical) deterministic finite automata.  The basic
computation step is to compute the next state of each processor
according to its state and the states of its neighbours.  In
\cite{Angluin} an asynchronous model is considered.  A basic
computation step means that two adjacent vertices exchange their
labels and then compute new ones.  In \cite{Yaka1} an asynchronous
model is studied where a basic computation step means that a
processor either changes its state and sends a message or it receives
a message.
In \cite{BV,BV0} networks are directed graphs coloured on their arcs;
each processor changes its state depending on its previous state and
on the states of its in-neighbours. Activation of processors may
be synchronous, asynchronous or interleaved.
In \cite{Stockmeyer} the aim is a study of distributed computations 
that can be done in a network within a time independent of the size of
the network.

\noindent
\subsection{The Termination Detection  Problem}
Starting with the works by Angluin \cite{Angluin} and Itai and Rodeh 
\cite{IR}, many papers have discussed the question: what functions 
can be computed by distributed algorithms in networks where knowledge 
about the network topology is limited?

Two important factors limiting the computational power of distributed 
systems are {\em symmetry} and {\em explicit termination}.
Some functions can be computed by an algorithm that terminates {\em 
implicitly} but not by an {\em explicitly} terminating algorithm.
In an implicitly terminating algorithm, each execution is finite and in 
the last state of the execution each node has the correct result.
However, the nodes are not aware that their state is the last one in the 
execution; with an explicitly terminating algorithm, nodes know the
local or global termination of the algorithm.

\subsubsection{Known Results about the Termination Detection Problem.}
Impossibility proofs for distributed computations quite often use the 
{\em replay} technique.
Starting from a (supposedly correct) execution of an algorithm, an 
execution is constructed in which the same steps are taken by nodes in a 
different network.
The mechanics of distributed execution dictate that this can happen, if 
the nodes are {\em locally} in the same situation, and this is precisely 
what is expressed by the existence of coverings.
The impossibility result implies that such awareness can never be 
obtained in a finite computation.
During the nineteen eighties there were many proposals for {\em 
termination detection} algorithms: such algorithms transform implicitly 
into explicitly terminating algorithms.
Several conditions were found to allow such algorithms (thus to null the 
difference between implicitly and explicitly computable functions) and 
for each of these conditions a specific algorithm was given 
(see \cite{Mat87,Lynch,Tel}).
These conditions include:
\begin{enumerate}
\item a unique {\em leader} exists in the network,
\item the network is known to be a tree,
\item the diameter of the network is known,
\item the nodes have different identification numbers.
\end{enumerate}
\subsubsection{The Main Result.}
In this paper we show that these four conditions are just special cases 
of one common criterion, namely that the local knowledge of nodes 
prohibits the existence of quasi-coverings of unbounded radius.
We also prove, by generalising the existing impossibility proofs to the 
limit, that in families with quasi-coverings of unbounded radius, 
termination detection is impossible. 
Informally, we prove (see Theorem \ref{caracOTD}):\par
{\it
  A distributed task $T=(\gfam,S)$ is locally computable with explicit
  termination detection if and only if 
  \begin{theoenum}
  \item $S$ is covering-lifting closed on \gfam,
  \item there exists a recursive function $r$
    such that for any $\bH$, there is no  strict
    quasi-covering  of \bH  of radius $r(\bH)$ in \gfam. 
  \end{theoenum}
}

Actually, we investigate different termination detection schemes:
local termination detection, observed termination detection and global
termination detection. This is  explained later in this introduction.

This is the first time, to our knowledge, that computability of a
distributed task (that is known to relate to ``local symmetries'') is
fully distinguished from the problem of detecting a kind of termination of a
distributed computation.

\subsubsection{Structural Knowledge and Labelled Graphs}
The definition of coverings and quasi-coverings are extended to 
include node and link labellings as well.
In the extension it is required that a node is mapped to a node with the 
same label, and links are mapped to links with the same label.
Our approach then naturally abstracts away the difference between 
anonymous or non-anonymous, centred or uniform
networks.
Indeed, the network being centred is modelled by considering as local 
knowledge that the graph family is the collection of graphs that contain 
exactly \emph{one} node with the label {\em leader}.

Specific assumptions (leader, identities, sense of direction, knowledge 
of size) now are examples of local knowledge that prevents certain 
quasi-coverings, thus allowing termination detection to take place.
Weak sense of direction (WSD) allows to distinguish closed from open 
walks, which is sufficiently strong to rule out all non-trivial 
quasi-coverings.
Thus termination detection is possible in all systems with WSD.

\subsection{The Election Problem}

As a very fundamental and illustrative problem, we investigate the
election problem. 
The election problem is one of the paradigms of the theory
of distributed computing. 
It was first posed by LeLann \cite{Lelann}.
Considering a network of processors 
the election problem is to arrive at a configuration where
exactly one processor is in the state ${\it {elected}}$ and all other
processors are in the state \textit{non-elected}.
The elected vertex is used to make decisions, to centralise or
to broadcast some information.
\subsubsection{Known Results about the Election Problem.}
Graphs where election is possible were already studied but the
algorithms usually involved some particular knowledge.
Solving the
problem for different knowledge has been investigated for some
particular cases  (see \cite{Attiya,Lynch,Tel} for details)
including:
\begin{enumerate}
\item the network is known to be a tree,
\item the network is known to be complete,
\item the network is known to be a grid or a torus,
\item the nodes have different identification numbers,
\item the network is known to be a ring and has a known
prime number of vertices.
\end{enumerate}
The classical proof techniques used for showing the non-existence of
election algorithm are based on coverings \cite{Angluin}, which is
a notion known from algebraic topology \cite{Massey}. 
A graph $G$ is a
covering of a graph $H$ if there is a surjective morphism from $G$ to
$H$ which is locally bijective.  The general idea,
used for impossibility proofs,  is as follows. If
$G$ and $H$ are two graphs such that $G$ covers $H$ and $G\neq H$, 
then every local
computation on $H$ induces a local computation on $G$ and
every label which appears in $H$ appears at least twice
in $G.$ Thus using $H$ it is always possible to build a computation
in $G$ 
such that the label $\it {elected}$ appears twice. By this way
it is proved that
there is no election algorithm for $G$  and $H$
(\cite{Angluin} Theorem 4.5).

A labelling is said to be locally bijective if vertices with the same
label are not in the same ball and 
have isomorphic labelled neighbourhoods.
A graph $G$ is non-ambiguous if any locally bijective labelling is
bijective.  Mazurkiewicz  has proved that,
knowing the size of graphs,
there exists an election algorithm for the class of non-ambiguous graphs
 \cite{MazurEnum}. This distributed 
algorithm, applied to a graph of size $n,$  assigns bijectively
numbers
of $[1..n]$ to vertices of $G.$ The elected vertex is the vertex
having the number $1.$

In \cite{MMW} the notion of quasi-covering has been introduced 
to study the problem of termination detection. A graph $G$
is a quasi-covering of a graph $H$ if $G$  is locally a covering
of $H$ (locally means that there is a vertex $v$ of
$G$ and a positive integer $k$ such that the  ball centred on $v$ 
of radius $k$ is a covering of a ball of $H$).
\subsubsection{The Main Result.}
We characterise which knowledge
is necessary and sufficient to have an election algorithm, or
equivalently, what is the general condition for a class of graphs to
admit an election algorithm: Theorem \ref{election}.
Sufficient conditions given below are just
special cases of criteria of Theorem \ref{election}.

We explain new parts in this theorem.
It is well known (see above) that the existence of an election
algorithm needs graphs minimal for the covering relation. We prove
in this paper 
that if a graph  is minimal for the covering relation
and admits quasi-coverings of arbitrary large radius in
the family there is no election algorithm. 
This part can be illustrated by the family of prime rings.
Indeed, prime rings are minimal for the covering relation nevertheless
there is no election algorithm for this family: without the knowledge
of the size, a ring admits quasi-covering prime rings of arbitrary
large radius.

These two results prove
one sense of Theorem \ref{election}.
To prove the converse:
\begin{itemize}
\item We remark that non-ambiguous graphs are precisely
graphs which are minimal for the covering relation.
\item We extend the Mazurkiewicz algorithm to labelled graphs.
\item We prove that the Mazurkiewicz algorithm applied in a labelled graph
$\mathbf G$ enables the ``cartography'', on each
node of $\mathbf G$, of a labelled graph $\mathbf H$ such that 
$\mathbf G$ is a quasi-covering of $\mathbf H;$
and when the computation is terminated $\mathbf G$ is a covering of
$\mathbf H.$
\item We define and 
we use an extension of an algorithm by Szymanski, Shi and Prywes
\cite{SSP} which enables the distributed
detection of stable properties in a labelled graph.
\item  We prove that the boundedness of the radius  of quasi-coverings 
of a given labelled graph enables to
each node $v$ to detect the termination  of the Mazurkiewicz algorithm 
and finally each node can decide if it is elected by testing if it has
obtained number $1$ by the Mazurkiewicz algorithm. 

\end{itemize}

\subsection{Tools}
\subsubsection{Coverings, Computations and Symmetry Breaking.}

The first step of a node in a distributed computation depends only on 
local initial knowledge of this node; only after receiving information 
from neighbours, the steps may depend on initial knowledge of these 
neighbours.
(Here initial knowledge includes the node's input, topological 
knowledge, degree, etc.)
Thus, consider a labelled graph $\mathbf G$ 
that contains a node $v$ with initial 
knowledge $x$, executing a distributed algorithm ${\cal A}$.
If $\mathbf G$ contains another node, $w$ say, with the same initial knowledge,
or a different labelled graph $\mathbf H$ 
contains a node with this knowledge, these 
nodes may thus execute the same first step if ${\cal A}$ is executed.
Now let $v$ in $\mathbf G$ have neighbours with 
initial knowledge $a$, $b$, and 
$c$ and assume that in the labelled graph $\mathbf H$, node $w$ 
also has neighbours with 
initial knowledge $a$, $b$, and $c$.
We thus create a ``local similarity'' to $\mathbf G$ of, in this case, 
a radius  1.
In this situation, not only will node $w$ start with the same step as 
node $v$, but also will receive the same information after the first 
step, and consequently will also perform the same second step.

Distributed tasks like election, enumeration (assigning different 
numbers to the nodes), and mutual exclusion require the network to reach 
a {\em non-symmetric} state.
A network state is symmetric if it contains different nodes that are in 
exactly the same situation; not only their local states, but also the 
states of their neighbours, of their neighbours' neighbours, etc.
That is, there exists a ``local similarity'' between different nodes of 
infinite radius.

The replay argument shows that different nodes that are locally similar
with infinite radius will exhibit the same behaviour in some infinite 
computation.
Thus, there is no algorithm that guarantees that the symmetry ceases in 
all finite computations.
Symmetry could be broken only by randomised protocols.

It is not difficult to see that local similarity of infinite radius may 
exist in finite graphs.
The classical example is a ring $R_6$ of six nodes, with initial states 
$a$, $b$, $c$, $a$, $b$, $c$.
Indeed, the two nodes with state $a$ both have neighbours in state $b$ 
and $c$, and so on, so the local similarity exists over an infinite 
radius.

The ring $R_6$ can be mapped into a ring $R_3$ with only three
nodes, with initial states $a$, $b$, and $c$, in such a way that each 
node is mapped to a node with the same initial state {\em and with the 
same states in neighbours}.
Such a mapping is called a {\em covering} and is the mathematical tool 
to prove the existence of symmetries.

\subsubsection{The Mazurkiewicz Algorithm.}

The proofs of our results used the fundamental Mazurkiewicz
distributed enumeration  algorithm.
A distributed enumeration algorithm on a graph $G$ is a distributed
algorithm such that the result of any computation is a labelling of
the vertices that is a bijection from $V(G)$ to
$\{1,2,\dots,|V(G)|\}$.  In  \cite{MazurEnum}, Mazurkiewicz 
presents a distributed enumeration algorithm for the class of
non-ambiguous graphs (graphs such that any local bijective labelling
is a bijective labelling). In this paper we prove that
the family of non-ambiguous graphs is the family of graphs
minimal for the covering relation. 

We prove also that
a run of the Mazurkiewicz algorithm on a labelled graph $\mathbf G$ 
(not necessarily
minimal for the covering relation) enables the  computation
on each vertex of $\mathbf G$ 
of a graph $\mathbf H$ quasi-covered by $\mathbf G$ 
(the quasi-covering becomes a
covering when the algorithm halts): we obtain a universal 
algorithm.

\subsubsection{The Szymanski, Shy and Prywes Algorithm and
Quasi-Co\-verings Relate to Termination Detection.}
Termination detection requires that a node certifies, in a {\em finite} 
computation, that all nodes of the network have completed their 
computation.
However, in a finite computation only information about a bounded region 
in the network can be gathered.
The algorithm by Szymanski, Shy, and Prywes does this for a region of 
pre-specified diameter; the assumption is necessary that the diameter of 
the {\em entire network} is known.
This implies that, termination detection, unlike symmetry breaking, is 
possible in  {\em every graph}, but provided some {\em knowledge}.

Network knowledge in an algorithm is modelled by a graph {\em family} in 
which the algorithm is required to work.
The detection algorithm by Szymanski {\em et al.}\ can be generalised in 
this way to work in a labelled graph family ${\cal F}$.
Nodes observe their neighbourhood and determine in what labelled graph 
$\mathbf H$ of  ${\cal F}$ they are.
Then they try to get a bound $k$ on the radius to which a different 
labelled 
graph of ${\cal F}$ can be locally similar to $\mathbf H$, 
and then certify that 
all nodes within distance $k$ are completed.
The universal termination detection algorithm thus combines the
Mazurkiewicz algorithm with (minimal) topological knowledge 
\cite{MazurEnum} and a known termination detection algorithm.

Of course the approach fails if a labelled graph 
${\mathbf H} \in {\cal F}$ is locally 
similar, with {\em unbounded} radius, to other graphs in ${\cal F}$.
Local similarities of this type are made precise in the notion of {\em 
quasi-coverings}.
Fortunately, the impossibility proofs for termination detection can be 
extended to cover exactly those families of labelled graphs that contain such 
unbounded-radius coverings.
Consequently, the sketched universal termination detection algorithm is 
the most general algorithm possible.

\subsubsection{Other Termination Detections}

In fact, in the previous algorithm, what is detected is that all
output values are correctly computed: the task is terminated, the
distributed algorithm is not terminated. Indeed, without symmetry
breaking conditions, we cannot detect the end of the
algorithm. Given a symmetric network, the ``last'' step can be
performed on at least two nodes. We call this kind of
detection {\em observed termination detection} because in this case,
the algorithm acts as an ``observer'' that knows when the underlying
computation of values is finished.
We do not ask this observer algorithm to detect its {\em own}
termination. Thus we distinguished the detection of the global
termination of the task from the detection of the termination of the
detection... This is presented in Theorem~\ref{caracOTD}.

In order to precise what can be explicit termination, we define
also  other kinds of termination detection: {\em detection of the  local
termination} (the nodes know when they enter their final step) and
{\em global termination detection} (one node knows that the
distributed   algorithm is finished). This last termination detection
scheme is characterised in Theorem \ref{caracGTD} that adds classical
coverings-based symmetry breaking conditions to the characterisation
of observed termination detection.

Such refinements of the notion
of termination  of a distributed algorithm are necessary to address
all kind of termination that are encountered in distributed
computing. One can think in particular about the composition of
distributed algorithms where observed termination detection seems not
enough decentralised.  

For example, from Th.~\ref{caracOTD}, it can be shown that they are no
distributed algorithm with detection of the global termination for
such computations - that are usually preliminary to general
distributed tasks - like computing the degree of a node, or any
computations that involve only a local part of the network (like in
\cite{Stockmeyer}). Indeed, on  
a huge network, without knowledge of something like a bound of the
diameter, a node can not even know if a very distant node has ever
started the distributed algorithm. Theorem~\ref{caracLTD} gives a
characterisation when the task is uniform, \ie,{} when the same value has
to be computed everywhere in the network. Open problems remains for
this kind of termination detection.

Finally, we show that, as it seems intuitively, these notions form a
strict hierarchy.

\subsection{Summary}
Section 2 reviews the definitions of coverings and quasi-coverings.
It pres\-e\-n\-ts local computations and their  relations with  coverings
and quasi-coverings.  
Section 3 presents local computations, coverings,
quasi-cov\-e\-r\-i\-n\-gs 
with their
properties that we need in the sequel of the paper.
Section 4 is 
devoted to the Mazurkiewicz algorithm,
Szymanski, Shy and Prywes algorithm and some extensions. 
In Section 5, we define formally our four notions of termination
detection (no detection, local termination, observed termination,
global termination) and gives numerous examples.
Our main results concerning the termination detection problem and the
 election problem are formulated and proved 
in Section 6 and Section 8. Section 7 presents some applications of
the theorems that present classical network hypothesises as corollaries.

This paper is an extended and improved version of the extended
abstracts \cite{MT} 
(the termination problem) and \cite{GMelection} (the election problem).

\section{Basic Notions and Notations}

\subsection{Graphs}
The notations used here are essentially standard \cite{Rosenpress}.  We
only consider finite, undirected, connected graphs without
multiple edges and self-loop.  
If $G$ is a graph, then $V(G)$ denotes the set of vertices and $E(G)$
denotes the set of edges. Two vertices $u$ and $v$  are said to be
adjacent if
$\{u,v\}$ belongs to $E(G).$
The distance between two vertices $u,v$ is denoted $d(u,v)$. 
The set of neighbours of $v$ in $G,$ denoted $N_G(v),$ is the set
of all vertices of $G$ adjacent to $v.$
For a vertex $v,$
 we denote by $B_G(v)$ the ball of radius 1 with center
$v,$ that is the graph with vertices $N_G(v) \cup \{v\}$ 
and edges $\left\{\{u,v\}\in  E(G) \mid u\in V(G) \right\}.$
We also denote by $B_G(v,r)$ the ball of center $v$ and radius
$r\in\N$. 

A homomorphism between $G$ and $H$ is a mapping $\gamma \colon V(G)\ra V(H)$ such that if $\{u,v\}$ is an
edge of $G$ then $\{\gamma(u),\gamma(v)\}$ is an edge of $H$.
Since
we deal only with graphs without self-loop, we have
$\gamma(u)\not=\gamma(v)$ whenever $\{u,v\}$ is an edge of $G$.  Note also
that $\gamma(N_G(u))\subseteq N_H(\gamma(u)).$
For an edge $\{u,v\}$ of $G$ we define 
$\gamma(\{u,v\})=\{\gamma(u),\gamma(v)\};$ this extends
$\gamma$ to a mapping $V(G)\cup E(G)\ra V(H) \cup E(H).$
 We say that $\gamma$
is an isomorphism if $\gamma$ is bijective and $\gamma^{-1}$ is 
a homomorphism, too.  We write $G\simeq G'$ whenever  $G$ and $G'$ are
isomorphic.  A class of graphs will be any set of graphs containing all graphs isomorphic to some of its
elements. The class of all graphs will be denoted $\mathcal G$.

For any set $S$, $\card(S)$ denotes the cardinality of $S$. For any
integer $q$, we denote by $[1,q]$ the set $\{1,2,\dots,q\}.$

\subsection{Labelled Graphs}
Throughout the paper we will consider  graphs where
vertices and edges are labelled with labels from a recursive
alphabet $L$. A graph labelled over $L$ will be denoted by
$(G,\lambda)$, where $G$ is a graph and $\lambda \colon V(G) \cup
E(G)\ra L$
is the labelling function. 
The graph $G$ is called the underlying
graph and the mapping $\lambda$ is a labelling of $G$.  
For a labelled graph $(G,\lambda)$, $lab((G,\lambda))$ is the set of labels
that occur in $(G,\lambda),$ i.e.,
$$
lab((G,\lambda))=\{ \lambda(v) | v\in V(G)\}.
$$
The class of
labelled graphs over some fixed alphabet $L$ will be denoted by 
${{\cal G}_L}$. 
Note that since $L$ is recursive, also ${{\cal G}_L}$ is recursive. 

Let $(G,\lambda)$ and $(G',\lambda')$ be two labelled graphs. Then
$(G,\lambda)$ is a subgraph of $(G',\lambda')$, denoted by
$(G,\lambda) \subseteq (G',\lambda')$, if $G$ is a subgraph of $G'$
and $\lambda$ is the restriction of the labelling $\lambda'$ to
$V(G)\cup E(G)$.

A mapping $ \gamma\colon V(G) \ra V(G')$ is a
homomorphism from $(G,\lambda)$ to $(G',\lambda')$ if $\gamma$ is a
graph homomorphism from $G$ to $G'$ which preserves the labelling,
i.e., such that $\lambda'(\gamma(x))=\lambda(x)$ holds for every 
$x\in V(G) \cup E(G).$ 

An {\em occurrence} of $(G,\lambda)$ in $(G',\lambda')$ is an
isomorphism $\gamma$ between $(G,\lambda)$ and a subgraph $(H,\eta )$
of $(G',\lambda')$. It shall be denoted 
$\gamma: (G,\lambda)\hookrightarrow(G',\lambda').$

Labelled graphs will be designated by  bold letters like 
$\mathbf G,$ $\mathbf H,\ldots$
If $\mathbf G$ is a labelled graph, then $G$ denotes  the underlying
graph. 

\subsection{Coverings}
We say that a graph $G$ is a {\em covering} of a graph $H$ via $\gamma$
if $\gamma$ is  a surjective 
homomorphism  from $G$ onto $H$ such that for every vertex
$v$ of $V(G)$ the restriction of $\gamma$ to $B_G(v)$ is a
bijection onto $B_H(\gamma(v))$.  The covering is proper if $G$ and
$H$ are not isomorphic. 

Examples and properties of coverings linked to networks are presented 
in \cite{Leeuwen,Bodlaender}. A generalization of coverings called fibrations
has been studied by Boldi and Vigna in \cite{BVfibrations}, this paper
emphasizes properties which found applications in distributed computing.

\begin{example} \label{ex_cov}
Let $R_n$, $n>2$, denote the ring on $n$ vertices defined
by $V(R_n)=[0,n-1]$ and $E(R_n)=\{\{x,y\} \mid y=x+1$ (mod $n$)$\}$.
Let now $m\ge n$ and $\gamma_{m,n}:[0,m]\longrightarrow[0,n]$ be
the mapping defined by $\gamma_{m,n}(i)=i$ (mod $n$), for every $i\in[0,m]$.
It is  easy to check that for every $n>2$ and for every $k>2,$
the ring $R_{k\times n}$ is a covering of the ring $R_n$ via
the mapping $\gamma_{k\times n,n}$.
\end{example}

The notion of covering extends  to labelled graphs in an obvious
way.  The labelled graph $(H,\lambda')$ is {\em covered} by
$(G,\lambda)$ via $\gamma,$  if $\gamma$ preserves labels and is
a covering from $G$ to $H$.

A graph $G$ is called {\em covering-minimal} if every
covering from $G$ to some $H$ is a bijection.
Note that a graph covering is exactly a covering in the classical
sense of algebraic topology, see \cite{Massey}. We have  the following basic property of coverings
\cite{Reidemeister}:

\begin{lemma}
  \label{un}
  Let $G$ be a covering of $H$ via $\gamma$ and let $v_1,v_2\in
  V(G)$ be such that $v_1\neq v_2$.  If $\gamma(v_1)=\gamma(v_2)$
  then $B_G(v_1) \cap B_G(v_2)=\emptyset$.
\end{lemma}

\begin{lemma}\label{arbre}
  Suppose that $G$ is a covering of $H$ via $\gamma.$ Let $T$ be
  a subgraph of $H.$ If $T$ is a tree then $\gamma^{-1}(T)$ is a set
  of disjoint trees, each isomorphic to $T.$
\end{lemma}
By considering simple paths between any two vertices,
the previous lemma implies:

\begin{lemma} \label{sheets}
  For every  covering  $\gamma$ from $G$ to $H$ there exists an integer $q$ such that $\card(\gamma^{-1}(v))=q$, for
  all $v\in V(H).$  
\end{lemma}

 The integer $q$ in the previous lemma is called the number of {\em sheets} of the covering.
We also refer to $\gamma$ as a {\em $q$-sheeted covering}.

\begin{figure}[hbtp]
  \begin{minipage}[c]{0.45\textwidth}%
    \begin{example}\label{1stex}
      A simple example of a 2-sheeted covering  is given in
      Fig.~\ref{fig:exemple}. The image of each vertex of $G$ is given
      by the corresponding Roman letter. Furthermore, we note
      that the image of each vertex is also given by its position on
      the $H$ pattern (the spanning tree of $H$ suggested in the figure).
      All  examples of coverings below will be implicitely described
      by this geometric scheme, that is based  on Theorem \ref{reid}.
    \end{example}
  \end{minipage}
  {\hskip1cm} 
  \begin{minipage}[c]{0.45\textwidth}
    \begin{center}
      \input{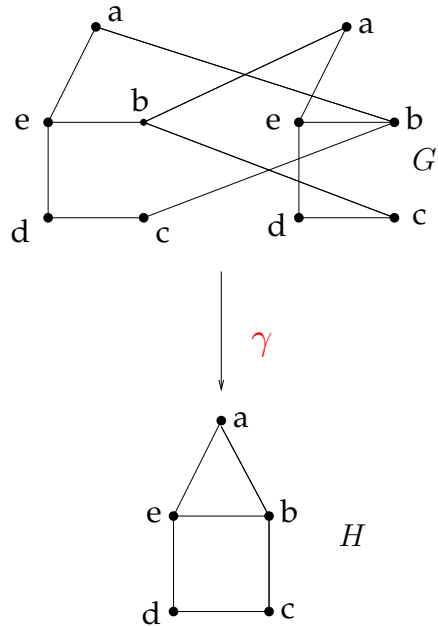}
      \caption{The morphism $\gamma$ is a covering from $G$ to
        $H$.}%
      \label{fig:exemple}
    \end{center}
  \end{minipage}
\end{figure}

Note also that for the rings $R_{k\times n}$ and $R_n$ 
the number of sheets is $k.$

In \cite{Reidemeister}, it is shown that all coverings of
$H$ can be obtained from a given spanning tree of $H$:
\begin{theorem}[\cite{Reidemeister}]\label{reid}
  Let $H$ be a graph and $T$ a spanning tree of $H$. A connected graph $G$ is
  a covering of $H$ if and only if there exist a non-negative integer
  $q$ and a set 
  $\Sigma=\{\sigma_{(x,y)}\mid\ x,y\in V(H), \{x,y\}\in E(H)\setminus
  E(T)\}$ of
  permutations\footnote { with the convention that $\sigma_{(x,y)} =
    \sigma^{-1}_{(y,x)}$ }  on $[1,q]$ such that $G$ is isomorphic to the graph
  $H_{T,\Sigma}$ defined by:
  \begin{eqnarray*}
    V(H_{T,\Sigma}) & = & \{ (x,i)\mid\ x\in V(H)\mid i\in [1,q] \},\\
    E(H_{T,\Sigma}) & = & \{\ \{(x,i),(y,i)\}\mid \{x,y\}\in E(T),\ i\in [1,q] \}\ \cup\\
    & & \{\ \{(x,i),(y,\sigma_{(x,y)}(i))\}\mid  \{x,y\}\in
    E(H)\setminus E(T),\ i\in [1,q] \}. 
  \end{eqnarray*}
\end{theorem}

\subsubsection{The Universal Covering.}
The universal covering of a graph 
is a special example of covering. It may be
defined as follows \cite{Angluin,Leighton}. Let $G$ be a graph, 
let $v$ be a vertex of 
$G,$  the universal covering of $G,$ denoted $U(G),$ is the infinite tree whose
vertex set is the set of all finite walks from $v$ in $G$ that do not traverse
the same edge in two consecutive steps. Two nodes are adjacent if
one is a one-step extension of the other. It is easy to verify
that $U(G)$ is a tree, unique up to isomorphism and independent of the choice
of $v.$
Clearly $U(G)$ covers $G.$ See Section\ref{univcovering} for a more formal definition. 
\subsection{Ambiguous Graphs and Coverings}
In this part we give the definition of ambiguous graphs introduced
by Mazurkiewicz in \cite{MazurEnum} and we show that the 
non-ambiguous graphs are precisely the covering-minimal graphs.

A labelling is said to be locally bijective if vertices with the same
label are not in the same ball and 
have isomorphic labelled neighbourhoods. Formally, we have:
\begin{definition}\cite{MazurEnum}\label{locbij}
  Let $L$ be a set of labels and let $(G,\lambda)$ be a labelled
  graph. The labelling $\lambda$ is  {\em locally bijective}
  if it verifies the following two conditions:
  \begin{enumerate}
  \item \label{inj_amb} For each $v\in V$ and for all $v',v'' \in
    B_G(v)$ we have $\lambda(v')=\lambda(v'')$ if and only if
    $v'=v''$.
  \item \label{surj_amb} For all $v',v''\in V$ such that
    $\lambda(v')=\lambda(v'')$, the labelled balls
    $(B_G(v'),\lambda)$ and $(B_G(v''),\lambda)$ are isomorphic.
  \end{enumerate}
  A graph $G$ is  {\em ambiguous} if there exists a non-bijective
  labelling of $G$ which is locally bijective.
\end{definition}
The labelling of the graph $G$ in Figure 1 proves that $G$ is ambiguous.

Locally bijective labellings and coverings are closely related
through quotient graphs.
\begin{definition} Let $\lambda$ be a labelling of the graph $G$. We
  define the {\em quotient graph} $G/\lambda$ by letting:
  \begin{itemize}
  \item $V(G/\lambda)=\lambda(V(G)),$ and
  \item $E(G/\lambda)=\{\{\alpha,\alpha'\} \mid \exists v,v'\in
    V(G) \text{ such that } \{v,v'\} \in E(G),\alpha=\lambda(v),\alpha'=\lambda(v')\}.$
  \end{itemize}
\end{definition}

\begin{lemma}\label{LB_kcov}
Let $G$ be a graph:
  \begin{enumerate}
  \item If $\lambda$ is a locally bijective labelling of $G$ then
    the quotient mapping $G\longrightarrow G/\lambda$ is a covering.
  \item Every covering  $\gamma:G\longrightarrow H$ 
    defines a locally bijective labelling of $G.$
  \end{enumerate}
\end{lemma}

\begin{proof}
  \begin{enumerate}
  \item Using  condition
    (\ref{inj_amb}) in Definition \ref{locbij}
     we note that
    $G/\lambda$ has no self-loop. 
    Moreover,
    the conditions (\ref{inj_amb}) and (\ref{surj_amb}) imply
    that $\lambda$ is a bijection from $B_G(v)$ to
    $B_{G/\lambda}(\lambda(v))$,  for each $v\in V(G)$.
    Hence $B_G(v)$ and $B_{G/\lambda}(\lambda(v))$ are isomorphic.
  \item We consider $V(H)$ as set of labels and  we label a vertex 
 $v\in V(G)$   by $   \gamma(v).$ 
 It is straightforward to verify that this labelling is
    locally bijective.
  \end{enumerate}
\end{proof}

Using the previous lemma we obtain:
\begin{corollary}\label{ambiguous}
  A graph is non-ambiguous if and only if it is covering-minimal.
\end{corollary}

\section{Local Computations}
In this section we recall the definition of local computations
and their relation with coverings \cite{LMZrecog}. They model
distributed algorithms on networks of processors of arbitrary
topology. The network is 
represented as a connected, undirected graph where vertices denote
processors and edges denote direct communication links.  Labels (or
states) are attached to vertices and edges.

Graph relabelling systems and more generally local computations satisfy the
following constraints, that arise naturally when describing distributed
computations with  decentralized control:
\begin{enumerate}
\item[$(C1)$] they do not change the underlying graph 
  but only the
  labelling of its components (edges and/or vertices),
  the final labelling being the result of the
  computation,
\item[$(C2)$] they are {\em local}, that is, each relabelling step
  changes only a connected subgraph of a fixed size in the underlying graph,
\item[$(C3)$] they are {\em locally generated}, that is, the
  applicability of a relabelling rule
  only depends on the {\em local context} of the relabelled subgraph.
\end{enumerate}
The relabelling is performed until no more
transformation is possible, i.e., until a normal form is obtained.
\subsection{Local Computations}
Local computations as considered here can be described in the
following general framework.  Let ${\cal G}_L$ be the class of
$L$-labelled graphs and let $\RR \;\subseteq {\cal G}_L \times {\cal
  G}_L$  be a
binary relation on ${\cal G}_L$. Then $\RR$ is called a {\em graph rewriting
  relation}.  We assume that $\RR$ is closed under isomorphism, i.e.,
if ${\mathbf G} \RR {\mathbf{G}'}$ and ${\mathbf H}\simeq
{\mathbf G}$ then ${\mathbf{ H}} \RR {\mathbf H'}$ for some
labelled graph ${\mathbf H'} \simeq {\mathbf G'}$.  In the
remainder of the paper $\RR^*$ stands for the 
reflexive-transitive closure of $\RR.$ The labelled graph
${\mathbf G}$ is {\em  $\RR$-irreducible} if there is no $\mathbf G'$
such that 
$\mathbf G\RR {\mathbf G'}.$ For ${\mathbf G}\in {\cal G}_L,$
Irred$_{\RR}({\mathbf G})$ denotes the set of $\RR$-irreducible
(or just irreducible if $\RR$ is fixed) graphs obtained from
$\mathbf G$ using $\RR,$ i.e., 
$Irred_{\RR}({\mathbf G})=\{{\mathbf H} | {\mathbf G} {\RR}^*{\mathbf H}
\text{ and } {\mathbf H} \text{ is } {\RR}\text{-irreducible}\}.$

\begin{definition} 
  Let $\RR\;\subseteq {{\cal G}_L} \times {{\cal G}_L}$ 
  be a graph rewriting relation.
  \begin{enumerate}
  \item $\RR$ is a {\em relabelling relation} if whenever two labelled
    graphs are in relation then the underlying graphs are equal (we
    say equal, not just isomorphic), \ie,
    $${\mathbf G} \RR {\mathbf H}\text{  implies that } G=H.$$
  \item $\RR$ is {\em local} if it can only modify   balls of radius
    $1$, i.e., $(G,\lambda) \RR (G,\lambda')$ 
    implies that there exists a vertex $v\in V(G)$ such that
    $$
    \lambda(x)=\lambda'(x) \mbox{ for every } x \notin
    V(B_G(v))\cup E(B_G(v)).$$
  \end{enumerate}

  The labelled ball $(B_G(v),\lambda)$ is a
  {\em support} of the relabelling relation.
\end{definition}

The next definition states that a local relabelling relation $\RR$ is
{\em locally generated} if the applicability of any relabelling 
depends only on the balls of radius 1.

\begin{definition} \label{locgen}Let $\RR$ be a  relabelling
  relation. Then $\RR$ is 
  {\em locally generated} if it is local and the 
  following is satisfied. For all labelled graphs $(G,\lambda)$,
  $(G,\lambda')$, such that  $(G,\lambda) \RR (G,\lambda')$, 
  there exists a vertex $v\in V(G)$, such that 
  $\lambda(x)=\lambda'(x)$, for all $x \notin V(B_G(v)) \cup
    E(B_G(v)),$
  and such that for all $(H,\eta)$, $(H,\eta')$, $w \in V(H)$ 
  where the balls $B_G(v)$ and $B_H(w)$  are isomorphic via $\varphi
  \colon V(B_G(v)) \longrightarrow V(B_H(w))$ and $\varphi(v)=w$,
  the following conditions:
  \begin{enumerate}\romanitems
  \item $\lambda(x)=\eta(\varphi(x))$ and
    $\lambda'(x)=\eta'(\varphi(x))$ for all $ x \in V(B_G(v))\cup
    E(B_G(v)),$
  \item $\eta(x)=\eta'(x)$, for all $x \notin V(B_H(w)) \cup
    E(B_H(w)),$
  \end{enumerate}
  imply that $(H,\eta)\RR (H,\eta')$. 

\end{definition}

By definition, local computations on graphs are computations on graphs
corresponding to locally generated relabelling relations.

We only consider recursive relabelling relations such that
the set of irreducible graphs is recursive. The purpose of all 
assumptions about recursiveness  done throughout  the paper is to
have  ``reasonable'' objects w.r.t.~the computational
power. Furthermore, in order to prevent ambiguousness,
Turing-computability will only be addressed as ``recursivity'', and we
will restrict the use of the word ``computability''  to the 
context of local computations.

A  sequence $({\mathbf G}_i)_{0\leq i\leq n}$ is called an 
 $\RR$-{\it relabelling}
{\it sequence} (or {\it relabelling sequence,} when $\RR$ is clear from the
context) if ${\mathbf G}_i \RR {\mathbf G}_{i+1}$ for 
every $0\leq i < n$ (with $n$ being the length of the sequence). A
relabelling sequence of length $1$ is a {\it relabelling step.}
The relation $\RR$ is called {\em noetherian} on a graph $\mathbf G$
if there
is no infinite relabelling  sequence ${\mathbf G_0}\RR {\mathbf
  G_1}\RR \ldots,$ with ${\mathbf G}_0={\mathbf G}.$ 
The relation $\RR$ is noetherian on 
a set of graphs if it is noetherian on each graph of the set. Finally,
the relation $\RR$ is called noetherian if it is noetherian
on each graph.

\subsection{Graph Relabelling Systems}
We present now graph relabelling systems as used  for modelling
distributed algorithms, by describing the exact form of the
relabelling steps. Each step will modify a  
{\em star-graph}, i.e., a graph with a distinguished center vertex
connected to all other vertices (and having no other edge besides
these edges).  As any  ball of radius one is isomorphic to a  
labelled star-graph, the support (or precondition) 
of any  relabelling rule  will be supposed to
be a labelled star-graph.

\subsubsection{Graph Relabelling Rules.}\par\noindent  
A {\it graph relabelling rule}  is a triple
$r=(B_r,\lambda_r,\lambda_r')$, 
where $B_r$ is a star-graph and $\lambda_r$, $\lambda_r'$  are two
labellings of $B_r$. We refer to  $(B_r,\lambda)$ as the {\em
  precondition} of the rule $r$, whereas   
$(B_r,\lambda')$ is referred to as the {\em relabelling}  through $r$.

Let $r=(B_r,\lambda_r,\lambda_r')$ be a relabelling rule, $H$ an 
(unlabelled) graph and $\eta$, $\eta'$ two labellings of $H$. We say
that $(H,\eta')$ is obtained from $(H,\eta)$ by
applying the rule $r$ to  the occurrence $\varphi$ of
$B_r$ in $H$ (and we write $(H,\eta) {\underset {r,\varphi}
  {\Longrightarrow} } (H,\eta')$) 
if the following conditions are satisfied, with $v_0$ denoting the
center of $B_r$: 
\begin{enumerate}
\item $\varphi$ induces both an isomorphism from $(B_r,\lambda_r)$ to
  $B_{(H,\eta)}(\varphi(v_0))$ and from $(B_r,\lambda'_r)$ to
  $B_{(H,\eta')}(\varphi(v_0))$, 
\item $\eta'(x)=\eta(x)$  for all $x \in (V(H)\setminus
  V(B_H(\varphi(v_0))))\cup (E(H) \setminus E(B_H(\varphi(v_0)))$,  
\end{enumerate}%
In this case we also say  that $\varphi$ is an occurrence of the rule
$r$ in $(H,\eta)$ and the image of $B_r$ under $\varphi$ is called the
image of $r$ under $\varphi.$

The relabelling 
relation $\underset {r} \Longrightarrow$ induced by the rule $r$ is
defined by letting  
$(H,\eta) {\underset r  \Longrightarrow} (H,\eta')$
if there exists an occurrence $\varphi$ of $r$ in $(H,\eta)$ with
$(H,\eta) {\underset {r,\varphi}  \Longrightarrow} (H,\eta').$\par
Let $r=(B_r,\lambda_r,\lambda_r')$ and $s=(B_s,\lambda_s,\lambda_s')$
be two (not necessary distinct) relabelling rules and let 
$$\varphi_r:\quad (B_r,\lambda_r)\quad \hookrightarrow\quad (H,\eta), \quad
\varphi_s: \quad (B_s,\lambda_s) \quad\hookrightarrow\quad (H,\eta)$$
be two occurrences of $r$ and $s$ respectively in $(H,\eta)$. We say
that these two occurrences  {\it overlap} if\par
(i) the images of $B_r$ by $\varphi_r$ and $B_s$ by $\varphi_s$ have a
common vertex, and\par 
(ii) either $r\neq s$ or ($r=s$ and $\varphi_r\neq \varphi_s$).\par

\subsubsection{Graph Relabelling Systems.}  
A {\it graph relabelling system}  is a recursive set $R$
of graph relabelling rules, such that the set of labelled star-graphs that are
preconditions of a rule in $R$ is also recursive.

The relabelling relation $\underset R  \Longrightarrow$ is defined by 
$(G,\lambda)\quad {\underset R  \Longrightarrow}\quad (G,\lambda')$
if there is a rule  $r \in R$ such that
$(G,\lambda)\quad {\underset r  \Longrightarrow} \quad (G,\lambda').$

Examples of graph relabelling systems are presented in
\cite{LMS1,LMZrecog}.

Clearly, graph relabelling systems represent locally generated relabelling
relations. Conversely, any locally generated relabelling relation
can be represented by a graph relabelling system.

\begin{proposition}
  Let \grs  be a  relabelling relation.
  Then \grs is both locally generated and a recursive relation such that
  the set of irreducible graphs is recursive
  if and only if  there exists a graph
  relabelling system $R$ such that $\grs$ equals
  $\underset{R}\Longrightarrow.$ 
\end{proposition}

\begin{proof}
  Given a locally generated relabelling relation \grs, we have  to
  find a graph relabelling system $R$ that generates \grs.

  We define:
  \begin{eqnarray*}
    R &=& \{(B,\lambda,\lambda')\;\mid 
    \; B \mbox{ is a star-graph},\; (B,\lambda)\; \grs\; (B,\lambda')
    \}
  \end{eqnarray*}

First, $R$ is obviously recursive since \grs is. The set of
preconditions of $R$ is also recursive, since one can check wether a
precondition does not belong to the set of \grs-irreducible graphs.
It is then straightforward to  verify that $R$ generates 
exactly \grs from Definition~\ref{locgen}.
\end{proof}

In the following, we do not discriminate between a locally generated
relabelling relation and a graph relabelling system that generates
it. They, both, model distributed algorithms.

\subsubsection{Generic Rules.}

We explain here the convention under which we will
describe graph relabelling systems later. If the number of rules is
finite then we will describe all rules by their preconditions
 and relabellings. We will also
describe a family of rules by  a generic rule (``meta-rule''). 
In this case, we will consider a
generic star-graph of generic center $v_0$ and of generic set of
vertices $B(v_0)$. Within these conventions, we will refer 
 to a vertex $v$ of the star graph by writing 
$v\in B(v_0)$. 
If $\lambda(v)$ 
is the label of $v$ in the
precondition, then
$\lambda'(v)$  will be its label in the relabelling. We
will omit in the description labels that are not modified by the rule. This means that if $\lambda(v)$ is a 
label such that $\lambda'(v)$ is not explicitly described in the rule
for a given $v$, then  $\lambda'(v)=\lambda(v)$.
In all the examples of 
graph relabelling systems that we consider in this paper the edge labels are
never changed. 

We do not require relabelling rules to be antisymmetric, but obviously
a system with such rules would have some difficulties to terminate. Thus, in order to have
light preconditions for generic rules, we consider that a rule
(induced by a given generic rule) that would not modify any label in
the star-graph is not enabled. 

With these conventions, the only point we have to care about is to
verify that the set of graph relabelling rules and the set of preconditions
described by the generic rule are recursive.

\subsubsection{Example}
  {Our first example is a $(d+1)$-coloring  of regular graphs of 
    degree $d$.}
  \label{sec:colord}
  This example will allow us to use the above described conventions.
\begin{example}\label{colod}
  We consider the graph relabelling system
  $\mbox{\textsc{Colo}}_d$. The value of the label of a vertex
  $v$ is denoted by  $c(v)$. The ``colors'' used here are integers from
  $[1,d+1]$, all vertices are initially labeled by $0$. 
  The following generic rule means that if $v_0$ is labelled by 0, then
   $v_0$ is relabelled by the smallest value that does not occur as
    label of one of 
  its neighbours. The edge labels are not used in this example.
  \label{grs:dcolor}
  \begin{rrule}{$\mbox{\textsc{Colo}}_d$ }
    \ritem[$\mbox{\textsc{Colo}}_d$ ]{$(d+1)$-Coloring}{
    \item $c(v_0)=0$ 
      }{
    \item $c'(v_0):=\min\left([1,d+1]\setminus\{c(v)\mid v\in B(v_0),
        c(v)\neq 0\}\right)$ 
      }
  \end{rrule}
  The  figures below show an execution of   \textsc{Colo$_3$}.

  \newcommand{\sepcube}{
    \raisebox{2.2cm}{$\underset {\mbox{\textsc{Colo}}_3}  \Longrightarrow$ 
      }
    }

  \newcommand{\etiqinit}{$0$}

  \providecommand{\etiqa}{\etiqinit}
  \providecommand{\etiqb}{\etiqinit}
  \providecommand{\etiqc}{\etiqinit}
  \providecommand{\etiqd}{\etiqinit}
  \providecommand{\etiqe}{\etiqinit}
  \providecommand{\etiqf}{\etiqinit}
  \providecommand{\etiqg}{\etiqinit}
  \providecommand{\etiqh}{\etiqinit}

  The initial labelling is the following:
  \begin{center}
    \raisebox{0.25cm}{
      \input{cube.pstex_t}
      }
  \end{center}

  Two non-overlapping occurrences where a rule can be applied are indicated below:

  \begin{center}
    \input{cubeselect2.pstex_t}
  \end{center}
  
  A corresponding relabelling sequence is as below: 
  \noindent\begin{center}
    \input{cube.pstex_t}
    \sepcube
    \renewcommand{\etiqa}{{\bf 1}}
    \input{cube.pstex_t}

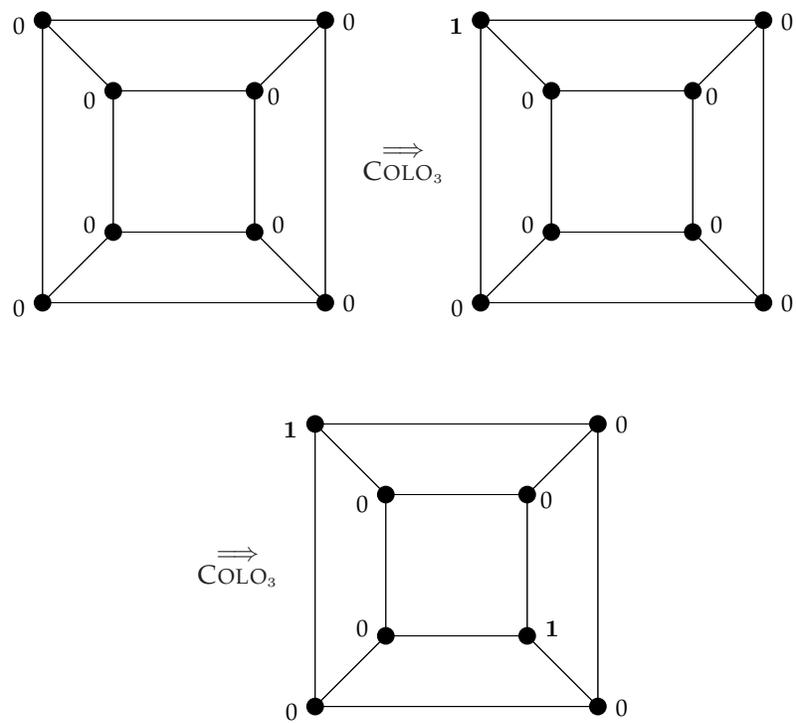
    \sepcube
    \renewcommand{\etiqg}{{\bf 1}}
    \input{cube.pstex_t}
  \end{center}
  \renewcommand{\etiqa}{{\bf 1}}
  \renewcommand{\etiqg}{{\bf 1}}

  The remaining part of the relabelling sequence is for instance: 

  \noindent
  \sepcube
  \renewcommand{\etiqh}{{\bf 2}}
  \input{cube.pstex_t}
  \sepcube 
  \renewcommand{\etiqd}{{\bf 3}}
  \input{cube.pstex_t}
  \sepcube
  \renewcommand{\etiqe}{{\bf 3}}
  \input{cube.pstex_t}
  \sepcube
  \renewcommand{\etiqb}{{\bf 2}}
  \input{cube.pstex_t}
  \sepcube
  \renewcommand{\etiqf}{{\bf 4}}
  \input{cube.pstex_t}
  \sepcube 
  \renewcommand{\etiqc}{{\bf 4}}
  \input{cube.pstex_t}
  \hspace{1cm}\strut

  One can note that the correctness of the algorithm follows
  from the fact that the set upon which the minimum is taken is never empty.
\end{example}

\subsection{Distributed Computations of Local Computations}
The notion of relabelling sequence defined above obviously
corresponds to a notion of {\em sequential} computation.  
Clearly, a locally generated relabelling relation  allows
parallel relabellings too, since non-overlapping balls may be relabelled
independently.  Thus we can define a distributed way of computing by
saying that two consecutive relabelling steps with disjoint supports
may be applied in any order (or concurrently).
More generally, any two relabelling sequences  such that  one
can be obtained from the other by 
exchanging successive concurrent steps, lead to the same result. 

Hence, our
notion of relabelling sequence
associated to a locally generated relabelling relation
may be regarded as a {\em
  serialization} \cite{Mazur1} of a distributed computation.  This
model is  asynchronous, in the sense that  several relabelling steps {\em may} be
done at the same time but we do not require that all of them have to
be performed.  In the sequel we will essentially handle sequential
relabelling sequences, but the reader should keep in mind that such
sequences may be done in parallel.

\subsection{Local Computations and Coverings}

We  now present the fundamental lemma connecting coverings and
locally generated  relabelling relations.  It states that whenever
$\mathbf G$ is a covering of $\mathbf H$, every relabelling step in 
$\mathbf H$ can
be lifted to a relabelling sequence in $\mathbf G$,  which is compatible with
the covering relation. It was first given in \cite{Angluin}.
\begin{lemma}[Lifting Lemma]\label{Rk_step}\label{lifting}
  Let $\RR$ be a locally generated relabelling relation and let
  $\mathbf G$ be a covering of $\mathbf H$  via $\gamma.$
  If  ${\mathbf H} \RR^* {\mathbf H'}$ then
  there exists $\mathbf G'$ such that ${\mathbf G}
  \RR^* {\mathbf G'}$ and $\mathbf G'$ is a covering of
  $\mathbf H'$ via $\gamma.$
\end{lemma}

\begin{proof}
  It suffices to show the claim for the case ${\mathbf H} \RR
  {\mathbf H'}$. Suppose that the relabelling step changes labels in
  $B_{H}(v),$  for some vertex $v \in V(H)$. We may apply this
  relabelling step to each of the disjoint labelled balls of
  $\gamma^{-1}(B_{H}(v))$, 
  since they are isomorphic to $B_{H}(v)$. This yields 
  $\mathbf G'$ which satisfies the claim.
\end{proof}  

 This is depicted  in the following commutative diagram:
\newcommand{\kco}{\operatorname{covering}}
\begin{equation*}
  \begin{CD}
    {\mathbf G} @>>\grs^*> {\mathbf G'}\\
    @V{\kco}VV  @VV{\kco}V\\
    {\mathbf H} @>>\grs^*> {\mathbf H'}
  \end{CD}
\end{equation*}

\subsection{Local Computations and Quasi-coverings}

We
will see now a configuration where only relabelling chains of bounded
length can be simulated. The notion of quasi-coverings was first
introduced in \cite{MMW} to prove impossibility of termination
detection in some cases. However
the definition of quasi-coverings here differs slightly from \cite{MMW},
providing new and simplified proofs, e.g., for Lemma
\ref{quasilifting} and 
Lemma
\ref{techlemmaGSSP}. Here, the key parameter is the radius and not the
size of the quasi-covering.

\begin{figure}[htbp]
\begin{center}
\input{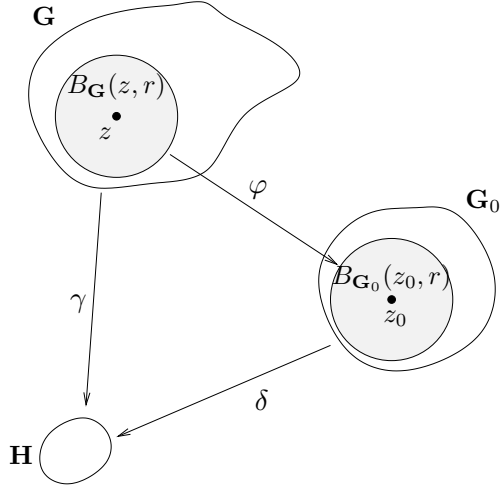}
\caption{\label{fig:qcov}$\gamma:\mathbf G\longrightarrow \bH$ is a
  quasi-covering of radius $r$ and associated covering
  $\delta:\bG_0\longrightarrow \bH$.}  
\end{center}
\end{figure}
\begin{definition} Let ${\mathbf G}, {\mathbf H}$ be two  
 labelled graphs and 
let $\gamma$ be a partial function  on $V(\bG)$ that assigns to each element of
a subset of $V(\bG)$ exactly one element of $V(\bH).$
   Then $\mathbf G$ is a {\em quasi-covering} of $\mathbf H$ 
via $\gamma$ of radius $r$ if there  exists a
  finite or infinite covering ${\mathbf G_0}$ of ${\mathbf H}$ via $\delta$, 
  vertices $z_0 \in V(G_0)$, $z \in V(G)$ such that:
\begin{enumerate}
\item $B_{\mathbf G}(z,r)$ is isomorphic via $\varphi$ to
    $B_{\mathbf G_0}(z_0,r)$,
\item the domain of definition of $\gamma$ contains $B_{G}(z,r),$ and 
\item $\gamma=\delta \circ \varphi$ when restricted to $V(B_{G}(z,r))$.
\end{enumerate}
$\card(V(B_{G}(z,r)))$ is called the {\em size} of the
  quasi-covering, and $z$ the {\em center}. The graph
  ${\mathbf G_0}$ is called the {\em associated covering} of the
  quasi-covering. See Figure~\ref{fig:qcov}.
\end{definition}
 Quasi-coverings have been introduced to study the problem of
the detection of the termination in \cite{MMW}.
The idea behind them is to enable the partial simulation of local
computations on a given graph  in a restricted area of a larger
graph. The restricted area where we can perform the simulation will
shrink while the number of simulated steps 
increases. 
The following lemma  makes precise how much the radius shrinks
when one step of simulation is performed:
\begin{lemma}[Quasi-Lifting Lemma]\label{quasilifting}
  Let $\RR$ be a locally generated relabelling relation and let
  $\mathbf G$ be a quasi-covering of $\mathbf H$  of radius $r$ via
  $\gamma.$ 
  Moreover, let ${\mathbf H} \RR {\mathbf H'}$.
  Then there exists $\mathbf G'$ such that ${\mathbf G}
  \RR^* {\mathbf G'}$ and $\mathbf G'$ is a quasi-covering of radius
  $r-2$ of $\mathbf H'.$
\end{lemma}
\begin{proof}
  Let $\bG_0$ be the associated covering and $z$ be the center of the
  ball of radius $r$. Suppose now the relabelling step ${\mathbf H}
  \RR {\mathbf H'}$ applies rule $R_0$ and modifies labels in $B_{\bH}(v),$
  for a given $v \in 
  V(\bH)$. The rule $R_0$ can also be applied to all the balls 
  $\delta^{-1}(B_{\bH}(v))$ yielding $\bG_0'$ and $\delta'$. It applied
  also to the balls  $\gamma^{-1}(B_{\bH}(v)))$  that are included in
  $B_{\bG}(z,r)$, since they are also isomorphic to $B_{\bH}(v)$. We get
  $\mathbf G'$ and $\gamma'$ satisfying the quasi-covering
  properties with radius $r-2$: consider $w$ in  $B_{\mathbf
  G'}(z,r-2)$: since any ball containing  $w$  is included
  in $B_{\bG}(z,r)$, $w$ and  $\gamma'(w)$ have the same label.   
\end{proof}

 This is depicted  in the following commutative diagram:
 \newcommand{\qco}{\operatorname{\mbox{
       \begin{tabular}[c]{c}
quasi-covering\\
of radius {$r$}
       \end{tabular}
}}}
 \newcommand{\qcoo}{\operatorname{\mbox{
       \begin{tabular}[c]{c}
quasi-covering\\
of radius {$r-2$}
       \end{tabular}
}}}
 \begin{equation*}
 \begin{CD}
\bG%
 @>>{\grs^*}> 
\bG'%
\\
 @V{\qco}VV @VV{\qcoo}V\\
\bH%
 @>>\grs>
\bH'%
\end{CD}
\end{equation*}

Using notation of this subsection:
\begin{definition}
  We define the {\em number of sheets $q$} of a quasi-covering to be the
  minimal cardinality of the sets of preimages 
  of vertices of $\mathbf  H$ which are in the ball:
 $$q = \min_{v\in V({\mathbf H})}|\{w \in\delta^{-1}(v)| 
             B_{{\mathbf K}}(w,1)\subset B_{K}(z_0,r)\}|.$$
\end{definition}

With this definition, the notion of number of sheets is equivalent in the case of coverings. 

\begin{definition}
  A quasi-covering is  {\em strict} if 
  $B_{\mathbf G}(z,r-1)\neq \mathbf G.$
\end{definition}

\begin{remark}\label{nonstrictqcovarecov}
  A non strict quasi-covering is simply a covering.
\end{remark}

\begin{remark}
With the same notation, if $\mathbf G$ is a strict quasi-covering of
$\mathbf H$ of radius $r$  then $|B_{\mathbf G}(z,r)|\geq r.$
\end{remark}
We have then the following technical lemma:
\begin{lemma}\label{techlemma}
  Let $\mathbf G$ be a strict quasi-covering of $\mathbf H$ of radius $r$ via
  $\gamma$. For any $q\in\N$, if $r \geq q|V(H)|$ then $\gamma$
  has at least $q$ sheets. 
\end{lemma}
\begin{proof}
  Note $\mathbf K$ the associated covering. 
  The quasi-covering being strict, we have that
  $|B_{\mathbf G}(z,r)|\geq r \geq q|V(H)|$, hence $|V(K)|\geq
  q|V(H)|$. We deduce from Lemma~\ref{sheets} that $\mathbf K$ has at least
  $q$ sheets.  

  Now, consider a spanning tree $T$ of $\mathbf H$ rooted on
  $\gamma(z)$. Note $T_1$ the lifting of $T$ rooted on $z_0$. By
  Theorem~\ref{reid}, there is $q-1$ distinct lifted spanning trees
  $T_2,\dots,T_q$ such that the subgraph induced by
  $T_1\cup\dots\cup T_q$ is connected.
  As $T$ has a diameter at most $|V(H)|-1$, we have
  that $T_1\cup\dots\cup T_q \subset B_{K}(z_0,q|V(H)|)$. That
  means that every vertex of $\mathbf H$ has at least $q$ preimages in
  $B_{K}(z_0,r)$, hence in $B_G(z,r)$.
\end{proof}

The following expresses a link of the radius and of the size of the
quasi-covering of a given graph.

\begin{lemma}\label{qcov_sizeandradius} Let $\mathbf H$ be a graph
  with maximal degree $d$. Then for all quasi-covering of $\mathbf H$
  of size $s$ and radius $r$, we have
  $$ s\leq (d+1)^{r} .$$
\end{lemma}
\begin{proof}
  Let $\mathbf G$ be a quasi-covering of $\mathbf H$. Let $z$ be the
  center, and $r$ the radius. $B_{\mathbf G}(z,r)$ is then a subgraph of
  maximal degree $d$. By induction, remarking that $|B(z,i+1)\backslash
  B(z,i)|\leq d|B(z,i)|$, we obtain that any ball of
  radius $r$ and maximal degree $d$ has a size at most $(d+1)^{r}$.
\end{proof}
This bound is obviously not optimal but sufficient for our
purpose. Remarking that a $q$-sheeted quasi-covering of a given graph
$\mathbf H$ has a size greater than $q|V(\mathbf H)|$, we get, from
these two lemmas, a complete relation between the radius and  the
number of sheets of a quasi-covering.

\subsection{Paths and Universal Coverings}
\label{univcovering} 
\newcommand{\wG}{{\ensuremath{\widehat{\bG}}}\xspace}
\newcommand{\wg}{{\ensuremath{\widehat{\pi}}}\xspace}

 A path is a sequence of neighbouring vertices in a graph.
\begin{definition}
  A {\em path} from   $u_0$ to $u_n$ in a  graph G is a sequence
  $\Gamma = (u_0, \dots, u_n)$ such   that for all $i$, 
  \begin{theoenum}
  \item $\{u_i,u_{i+1}\}\in E(G)$.
  \end{theoenum}

  Furthermore, if, for all $i$, 
  \begin{theoenum}
    \setcounter{theoenumcounter}{1}
      \item $u_{i-1}\neq u_{i+1}$ ,
  \end{theoenum} 
  we say that $\Gamma$ is a {\em non stuttering path}\cite{BVfibrations}.
\end{definition}
We denote by $\Gamma_G(u)$ the set of paths in $G$ starting from vertex
$u$.  
For any path $\Gamma=(u_0,\dots,u_n)$ and any vertex $v$, we note
$\Gamma v$ the path $(u_0,\dots,u_n,v)$.
\begin{definition}
  Let \bG be a (labelled) graph. Let $u$ be a vertex of \bG. We denote by $\wG(u)$
  the graph of non stuttering paths starting from $u$:
  \begin{eqnarray*}
    V(\wG(u))&=&\{\Gamma\in\Gamma_\bG(u)\mid \Gamma \mbox{ is non stuttering}\},\\
    E(\wG(u))&=&\{\{\Gamma,\Gamma'\}\mid \Gamma,\Gamma' \in V(\wG(u)),
                 \mbox{ and there exists a vertex $v$}\\
    &&\mbox{ of \bG such that } \Gamma'=\Gamma v\}.
  \end{eqnarray*}
  We denote by $\wg$ the projection of $\wG(u)$ on \bG that maps  any path to
  its final vertex.
\end{definition}
\newcommand{\boule}{B_\wG(z_\wG,3)}
\begin{figure}%
  \begin{center}
    \input{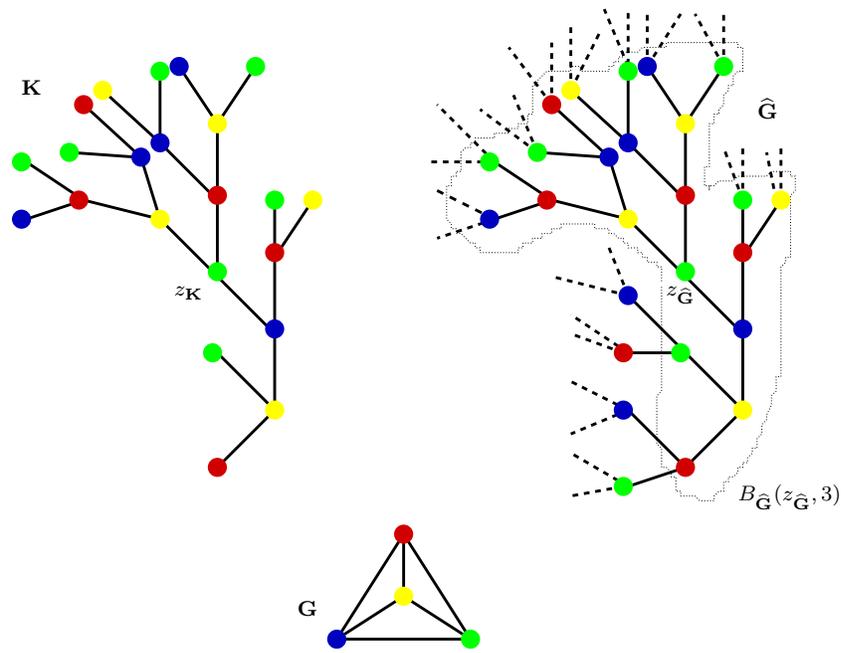}
    \caption{\bK is a quasi-covering of radius 3 of \bG, obtained by
      truncation of \wG}
    \label{fig:troncatureuniv}
  \end{center}
\end{figure}
\begin{proposition}
 The graph $\wG(u)$ is a covering of \bG via the projection \wg.
\end{proposition}
\begin{proof}
  Let $v$ a vertex of \bG.  Let a path 
  $\Gamma=(u_0,\dots,u_n)$ with $u_0=u$ and $u_n=v$. Suppose that
  $\Gamma$ is not the empty path. By construction,
  $\Gamma$ has as neighbours $(u_0,\dots,u_{n-1})$. Being non
  stuttering, it also has as neighbours the paths of the set
  $\{\Gamma w \mid
  w\in N(v), w\neq u_{n-1}\}$. Hence \wg defines an
  isomorphism from $B_{\wG(u)}(\Gamma)$ to $B_\bG(v)$. 
  If  $\Gamma$ is the empty path, the proof is obvious.
\end{proof}

For all vertices $u$,$v$, $\wG(u)$ is isomorphic to $\wG(v)$
\cite{BVfibrations}. We shall denote by \wG this graph defined up to
isomorphism.  We say that \wG is the {\em universal covering} of \bG.
\begin{remark}\label{troncatureuniv}
 This (possibly infinite) tree provides numerous examples of
 quasi-coverings. 
 For a given graph \bG, by truncation of the universal covering \wG to
 a ball of given radius, we obtain quasi-coverings of \bG. See
 Fig.~\ref{fig:troncatureuniv}. 

 The Reidemeister Theorem (Th.~\ref{reid}) is another tool to easily build
 quasi-covering of arbitrary radius.
\end{remark}

\subsection{Extension of Locally Generated Relabelling Relations}

In this subsection, we show how the properties of a graph relabelling
relation on a family \gfam can be naturally extended to the family of
graphs that are covered by a graph of \gfam.

\macro{\covgfam}{\widehat{\gfam}}

\begin{definition}
  Let \gfam be a graph family. We note \covgfam the family of graphs
  that are covered by a graph of \gfam. 
$$\covgfam=\{\bH\mid \exists\bG\in\gfam, \bG\mbox{ is a covering of }
\bH\}.$$ 
\end{definition}

Note that \gfam is a subset of \covgfam. The first easy property is
that if a \grs is noetherian on \gfam, it is also noetherian on
\covgfam.

\begin{lemma}\label{covgfamnoetherian}
  Let \grs be a relabelling system. If \grs is notetherian on \gfam, it
  is also noetherian on \covgfam.  
\end{lemma}

\begin{proof}
  Suppose there is an infinite relabelling chain on
  $\bH\in\covgfam$. Note \bG a graph in \gfam that is a covering of
  \bH. By the Lifing Lemma, we get an infinite relabelling chain on
  \bG. Hence a contradiction.
\end{proof}

\begin{remark}\label{covgfamnotrecursive}
  The closure under covering  of a recursive graph family is not
  necessarily recursive. Consider the following family 
  \begin{eqnarray*}
    \gfam_c&=\{G\mid &G \mbox{ is a ring and there exists } p,i,m\in\N
    \mbox{ such that }\\
    && p^m \mbox{ is the size of }G,\\
    && p \mbox{ is the $i$-th prime number},\\
    && \mbox{Turing Machine number $i$ has halted before step } m\}.
  \end{eqnarray*}
  The family $\gfam_c$ is obviously recursive and ${\covgfam_c}$ is 
  obviously non recursive: it is straighforward to see that deciding
  if a ring of prime size can be lifted in $\gfam_c$ corresponds to
  the Halting Problem for Turing Machines. 
\end{remark}

\section{Fundamental Algorithms}

In this section, we present our two fundamental algorithms.

\subsection{ Mazurkiewicz' Enumeration Algorithm}
\newcommand{\mk}{\ensuremath{{\mathcal{M}}}\xspace}
\newcommand{\NN}{\mathcal{N}}
\newcommand{\gse}{\rho(\mathbf G)}

A distributed enumeration algorithm on a graph $\mathbf G$ is a distributed
algorithm such that the result of any computation is a labelling of
the vertices that is a bijection from $V(G)$ to
$\{1,2,\dots,|V(G)|\}$. In particular, an enumeration of the vertices
where vertices know whether the algorithm has terminated solves the
election problem. In \cite{MazurEnum} Mazurkiewicz presents a distributed
enumeration algorithm for covering-minimal (non-ambiguous) graphs. 

The computation model in \cite{MazurEnum} consists
exactly in relabelling balls of radius $1$ and the 
initial  graph is unlabelled.

Mazurkiewicz' algorithm will be
denoted $\mk$. By abuse of language we still speak of an enumeration algorithm, 
even when it is applied to ambiguous graphs (for which
no enumeration algorithm exists, \cite{MazurEnum}). The  final labellings
that are incorrect 
from the enumeration point of view have interesting properties in the
 context of local computation. Namely, they determine a graph that is
 covered by the input graph.

In the following we describe 
Mazurkiewicz' algorithm including its extension to labelled
graphs.

\subsubsection{Enumeration Algorithm.}
We first give a general description of the algorithm $\mk$
applied to a graph $\mathbf G.$  
Let ${\mathbf G}=(G,\lambda)$ and consider  a vertex $v_0$ of $G,$ and
the set $\{v_1,...,v_d\}$ of neighbours of $v_0.$

The label of the vertex $v_0$ used by $\mk$
is the pair $(\lambda(v_0),c(v_0))$ 
where $c(v_0)$ is a triple 
$(n(v_0),N(v_0),M(v_0))$ representing the following information 
obtained during the computation (formal definitions are given below):
\begin{itemize}
\item $n(v_0) \in \N$ is the {\em number} of the vertex $v_0$ computed by the
  algorithm,
\item $N(v_0)\in\NN$ is the {\em local view} of $v_0,$ and it is either 
  empty or a family of triples defined by:
  $$
  \{(n(v_i),\lambda(v_i),\lambda(\{v_0,v_i\})) | 1\leq i \leq d\}
  ,$$
\item $M(v_0)\subseteq L\times \N\times \NN$ is the {\em mailbox} of $v_0$ and
  contains the whole information received by $v_0$ at any step of the
  computation. %
\end{itemize}
Each vertex $v$ attempts to get its
own number $n(v)$, which will be an integer between 1 and $|V(G)|$. A vertex
chooses a number and broadcasts it 
together with its label and its 
labelled neighbourhood all over the network. If a
vertex $u$ discovers the existence of another vertex $v$ with the same
number, then it compares its label and its 
local view, \ie, its number-labelled
ball, with the local view of its rival $v$. If the label of $v$ or the 
local view of $v$ is ``stronger'', then $u$ chooses another
number.  Each new number, with its local view, 
is broadcast again over the network.  At the
end of the computation it is not guaranteed that every vertex has a
unique number, unless the graph is covering-minimal.  However, all vertices
with the same number will have the same label and the same 
local view.%

\par

The crucial property of the algorithm is based on a total order on
local views such that the local  view of any vertex
cannot decrease during the computation. 
We assume for the rest of this paper that the set of labels $L$ is  totally ordered
by $<_L.$ 
Consider a  vertex $v_0$ with neighbourhood $\{v_1,...,v_d\}$ and assume that: 
\begin{itemize}
\item $n(v_1)\geq n(v_2)\geq ...\geq n(v_d),$
\item if $n(v_i)=n(v_{i+1})$ then $\lambda(v_i)\geq_L \lambda(v_{i+1}),$
\item if $n(v_i)=n(v_{i+1})$ and  $\lambda(v_i)= \lambda(v_{i+1})$
  then $\lambda(\{v_0,v_i\})\geq_L\lambda(\{v_0,v_{i+1}\}).$
\end{itemize}
Then the local view $N(v)$ is the $d$-tuple
$$((n(v_1),\lambda(v_1),\lambda(\{v_0,v_1\})),\dots,(n(v_d),\lambda(v_d),\lambda(\{v_0,v_d\}))).$$
Let $\NN_>$ be the set  of all such ordered tuples.
We define a total order
$\prec $ on $\NN_>$ by comparing the numbers, then the vertex labels 
and finally the edge labels. Formally, for two elements
$$((n_1,l_1,e_1),...,(n_d,l_d,e_d))\,
\mbox{and} \,((n_1',l_1',e_1'),...,(n_{d'}',l_{d'}',e_{d'}'))$$ 
of $\NN_>$
we define
$$
((n_1',l_1',e_1'),...,(n_{d'}',l_{d'}',e_{d'}'))
\prec 
((n_1,l_1,e_1),...,(n_d,l_d,e_d))
$$ if one of the following conditions holds:
\begin{enumerate}
\item $n_1=n'_1,...,n_{i-1}=n'_{i-1}$ and
  $n'_i < n_i$ for some $i$,
\item $d'<d$ and $n_1=n'_1,...,n_{d'}=n'_{d'}$,
\item $d=d'$, $n_1=n'_1,...,n_{d}=n'_{d}$ and 
  $l_1=l'_1,...,l_{i-1}=l'_{i-1}$ and
  $l_i'<_Ll_i$ for some $i$,
\item $d=d'$ and $n_1=n'_1,...,n_{d}=n'_{d}$ and
  $l_1=l'_1,...,l_{d}=l'_{d}$ and 
  $e_1=e'_1,...,e_{i-1}=e'_{i-1}$ and
  $e_i'<_Le_i$ for some $i$.
\end{enumerate}

If $N(u)\prec N(v)$, then we say that
the local view $N(v)$ of $v$ is stronger than the one of $u.$ The
order $\prec$ is a total order on $\NN=\NN_>\cup\{\emptyset\},$
with, by definition,: ${\emptyset}\prec N$ for every $N\in \NN_>.$

\medskip
We  now describe the algorithm through a graph
relabelling system.
The initial labelling of the vertex $v_0$ is 
$(\lambda(v_0),(0,\emptyset,\emptyset)).$

The rules are described below for a given ball
$B(v_0)$ with center $v_0$. The vertices $v$ of $B(v_0)$ have labels
$(\lambda(v),(n(v),N(v),M(v)))$. The labels obtained after applying a rule are
$(\lambda(v),(n'(v),N'(v),M'(v)))$.  
We recall that we omit labels that are 
unchanged. 

\begin{rrule}{\mk--}
  \ritem{Diffusion rule\label{dif_rule}}{
  \item There exists $v \in B(v_0)$ such that $M(v)\neq M(v_0)$.
    }{
  \item For all $v\in B(v_0)$, $M'(v):=\mathop{\bigcup}\limits_{w\in
      B(v_0)}M(w)$.
    }
  \ritem{Renaming rule\label{relab_rule}}{
  \item For all $v\in B(v_0), M(v)=M(v_0)$.
  \item $(n(v_0)=0)$ or \\$( n(v_0)>0
    \mbox{ and there exists }(l,n(v_0),N)\in M(v_0)
    \mbox{ such that }\\ (\lambda(v_0)< l)\mbox{ or }((\lambda(v_0)=l)
    \mbox{ and }  (N(v_0)\prec N)))$.
    }{
  \item $n'(v_0) = 1+\max \{n\in\N\mid (l,n,N)\in
    M(v_0)\,\,\text{for some}\,\, l,N \}$.
  \item For every $v\in B(v_0)$, $N'(v)$ is obtained from $N(v)$
    by replacing the value of $n(v_0)$  by $n'(v_0).$ 
  \item For every $v\in B(v_0),$ the mailbox contents $M(v)$
    changes to\\ 
    $M'(v) = M(v)\cup
    \{(\lambda(w),n'(w),N'(w)) | w \in B(v_0)\}$.
    }
\end{rrule}

\subsection{Properties of  Mazurkiewicz' Algorithm}

In order to make the paper self-contained, 
we present a complete proof of  the correctness of Mazurkiewicz'
algorithm in our framework following the  ideas developed in
\cite{MazurEnum}.

Let $\mathbf G$ be a labelled graph.
If $v$ is a vertex of $G$ then  the label of $v$ 
after a run $\rho$ of  Mazurkiewicz' algorithm
is denoted $(\lambda(v),c_{\rho}(v))$ with
$c_{\rho}(v)=(n_\rho(v),N_\rho(v),M_\rho(v))$ and $(\lambda,c_{\rho})$
denotes the final labelling.

\begin{theorem}\cite{MazurEnum}\label{fundamental}
  Any run $\rho$ 
  of   Mazurkiewicz' enumeration algorithm on a connected labelled graph
  ${\mathbf G}=(G,\lambda)$ terminates 
  and yields a final labelling $(\lambda,c_{\rho})$ verifying the
  following conditions for all vertices $v,v'$ of $G$:
  \begin{theoenum}
  \item Let $m$ be the maximal number in the final labelling, 
    $m=\max\limits_{v \in V(G)}
    n_\rho(v)$. Then for every $1 \le p \le m$ there is some $v \in
    V(G)$ with $n_\rho(v)=p$.
  \item \label{fondM} $M_\rho(v)=M_\rho(v')$.
  \item $(\lambda(v),n_\rho(v),N_\rho(v))\in M_\rho(v')$.
  \item \label{fondN} Let $(l,n,N) \in M_\rho(v')$. Then
    $\lambda(v)=l$, $n_\rho(v)=n$ and
    $N_\rho(v)=N$ for some vertex $v$ if and only if there is no pair
    $(l',n,N')\in M_\rho(v')$ with 
    $l<_Ll'$ or ($l=l'$ and $N \prec N'$).
  \item \label{fondn} $n_\rho(v) = n_\rho(v') \mbox{ implies } ( \lambda(v) =
    \lambda(v')$ and $ N(v) = N(v'))$ 
  \item $n_\rho$ induces a locally bijective labelling of $G$.
  \end{theoenum}
\end{theorem}

We first prove the following lemmas.
We say that a number $m$ is known by $v$ if $(l,m,N)\in M(v)$ for some
$l$ and some $N$. 
In the following $i$ is an integer denoting a computation step.
  Let $(\lambda(v),(n_i(v),N_i(v),$ $M_i(v))$ be the label of the vertex
  $v$   after the $i$th step of the computation. 
\begin{lemma} \label{croissance}
 For each $v,i:$
  \begin{itemize}
  \item $n_i(v)\leq n_{i+1}(v)$,
  \item $N_i(v)\preceq N_{i+1}(v)$,
  \item $M_i(v)\subseteq M_{i+1}(v)$.
  \end{itemize}
\end{lemma}
\begin{proof}
  The property is obviously true for the vertices that are not involved in
  the rule applied at step $i.$ For the other vertices we note that the
  {\em renaming rule} applied to $v_0$ increments $n_i(v_0)$, adds
  elements to some mailboxes and makes some $N(u)$ stronger. Moreover 
  the {\em   diffusion rule} only adds elements to mailboxes.

  The fact that $N_i(v)\preceq N_{i+1}(v)$ comes from the 
  definition of $\prec$. In other words, this order 
  ensures that the past local views of a vertex are always weaker
  than its present one.

  Furthermore, one of the
  inequalities is strict for at least one vertex, namely the one for
  which the previous  rule was applied. 
\end{proof}

\begin{lemma}  \label{knowing} For every $v\in V(G)$
  and $(l,m,N)\in M_i(v)$ there exists a vertex $w  \in V(G)$ such that
  $n_i(w) = m.$ 
\end{lemma}
\begin{proof}
  Assume that the number $m$ is known by $v$ and let $U=\{u\in V(G)\mid
  \exists j<i, n_j(u)=m\}.$  Obviously $U$ is not empty. Let $w\in U$ 
  and let $j<i$ such that  
  \begin{enumerate}
  \item $n_j(w)=m,$ 
  \item  for any $u\in U$ and for any $k<i$ verifying $n_k(u)=m$ we have:
    $N_k(u)\preceq N_j(w).$
  \end{enumerate}
  Clearly, the  {\em renaming rule} cannot be applied to $w$, hence $n_i(w)=m$.
\end{proof}

Next, we claim that whenever a number is known, all positive smaller
numbers are assigned to some vertex.
\begin{lemma} \label{knowingsmaller} For every vertex $v\in V(G)$ such that
  $n_i(v)\neq0$ and for every $ m\in [1, n_i(v)]$, 
  there exists some vertex $w  \in V(G)$ such that $n_i(w) = m.$
\end{lemma}
\begin{proof}
  We show this claim by  induction on $i$.
  At the initial step $(i=0)$ the assertion is true.
  Suppose that it holds for $i \ge 0$. If the {\em diffusion rule} is used,
  the assertion is true for $i+1$. If the {\em renaming rule} is applied to
  $v_0$ then we just have to verify it for $v_0$, and more precisely for
  all numbers  $m$ in the interval $\{n_i(v_0),n_i(v_0)+1,\dots,n_{i+1}(v_0)\}$.
  The property holds obviously for $n_{i+1}(v_0)$ and, being
  known by $v_0$ at step $i+1$, the property for $n_i(v_0)$
  is a consequence of Lemma~\ref{knowing}.

  If the interval $\{n_i(v_0)+1,\dots,n_{i+1}(v_0)-1\}$ is
  empty then the condition is obviously satisfied.
  Otherwise  by definition of the {\em renaming rule}, $n_{i+1}(v_0)-1$ is
  known by $v_0$ at step $i$ and thus Lemma \ref{knowing} implies that
  there exists $w\neq  v_0$ such that 
  $n_i(w)=n_{i+1}(v_0)-1$. For every $m\in
  \{n_i(v_0)+1,\dots,n_{i+1}(v_0)-1\}$, we have, by induction hypothesis
  on $w$  that there exists a vertex $x\in V(G)$ such that $n_i(x) =
  m$.  For every such $x$, because $v_0$ is the only vertex changing its
  name from step $i$ to $i+1$, $n_i(x)=n_{i+1}(x)$, which proves the
  assertion for step $i+1$. 
\end{proof}

We show now  Theorem \ref{fundamental}:
\begin{proof}

  As before, we denote by
  $(\lambda(v),(n_i(v),N_i(v),$ $M_i(v)))$  the label of the vertex
  $v$   after the $i$th step of the computation. 

  As there are no more than $|V(G)|$ different numbers assigned
  it follows from Lemma~\ref{croissance} and 
  from Lemma~\ref{knowingsmaller} that the algorithm terminates.

  The properties 1 to 6 of the final labelling are easily derived from the
  above part of the proof.
  \begin{enumerate}
  \item By Lemma~\ref{knowingsmaller} applied to the final labelling.
  \item Otherwise, the {\em diffusion rule} could be applied.
  \item A direct corollary of the previous property.
  \item We have obtained a final labelling, thus it is a direct 
    consequence of the diffusion rule and of the
    precondition of the renaming rule.
  \item A direct consequence of the previous point.
  \item The first part of Definition~\ref{locbij} is a consequence
    of the rewriting mechanism: when a vertex $v$ is numbered, its number
    is put in mailboxes of adjacent vertices. Thus vertices at distance
    $2$ of $v$ cannot have the same number as $v.$
    The second part of Definition~\ref{locbij} is a consequence
    of the precondition of the renaming rule: 
    the {\em renaming rule} could have been applied to
    vertices having the same number and non-isomorphic local
    views.
  \end{enumerate}
  This ends the proof of the theorem.
\end{proof}

\begin{remark}
  By points $1$ and $6$ of Theorem~\ref{fundamental}, and similarly
  to \cite{MazurEnum}, the algorithm computes  for non-ambiguous graphs
  (and thus for minimal graphs by Corollary \ref{ambiguous}),  a
  one-to-one correspondence $n_\rho$ between the set of 
  vertices of $G$ and the set of integers $\{1,\ldots,|V(G)|\}$.
\end{remark}

\subsection{Toward an Enhanced Mazurkiewicz' Algorithm} 
In this section we prove that even by applying Mazurkiewicz'
algorithm to a graph $\mathbf G$ that is not covering-minimal, we can get some relevant information. 
In this case, we prove that we can interpret the mailbox of
 the final labelling as a 
graph $\mathbf H$ that each vertex can compute and such that 
$\mathbf G$ is a covering of $\mathbf H.$

For a mailbox $M$, we define the graph of the ``strongest'' vertices as follows. 
\newcommand{\fort}{\ensuremath{\underline{\mbox{Strong}}}}
First, for $l\in L, n\in \N, N\in \NN, M\subseteq L \times  \N \times\NN$, we
define the predicate $\fort(l,n,N,M)$ that is true if  there 
is no $(l',n,N')\in M$ verifying
$$  
            l'>l \mbox{ or } (l=l' \mbox{ and }
             N\prec N')
              .
$$

The graph  $H_{M}$ of strongest vertices of $M$ is then defined by
\begin{eqnarray*}
V({H_M})&=&\{n\mid \exists N,l: \fort(l,n,N,M)\},\\
E({H_M})&=&\{\{n,n'\}\mid \exists N,l: \fort(l,n,N,M), \mbox{ and } 
           \exists l',l'': \\
&& \;\;\;\;\;N=(...,(n',l',l''),...) \;\}.
\end{eqnarray*}
We also define a labelling on this graph by
 $\lambda_M(n)=(n,l,N,M),$ with $\fort$$(n,l,N,M)$ for some $N,$ 
and $\lambda_M(\{n,n'\})=l'',$ with 
$\fort(n,l,N,M)$ and  $N=(...,(n',l',l''),...).$

The uniqueness of this
 definition comes from the definition of {\fort} and from 
 Theorem~\ref{fondn}.

 Let $\rho$ be a run of $\mk.$ Then  
$(H_{M_{\rho}(u)},$ $\lambda_{M_{\rho}(u)})$ does not 
depend on $u$ by Theorem~\ref{fundamental}.2. We then define
$\gse=(H_{M_{\rho}(u)},\lambda_{M_{\rho}(u)}),$ for any vertex
$u$. 
Finally, we have:

\begin{proposition}
For a given execution $\rho$ of Mazurkiewicz algorithm, we have
$$
V(\rho(\mathbf G))= \{n_\rho(v)| v\in V(G)\},
$$
$$
E(\rho(\mathbf G))=\{\{n_\rho(v),n_\rho(w)\} | \{v,w\}\in E(G)\}.
$$
\end{proposition}

\begin{remark} Before we emphasize the role of $\gse$, note  that
  $\gse$ can be locally computed by every vertex, and that the graph
  depends only on the label $M_\rho$.
\end{remark}

The next proposition  states
that we can see a run of $\mk$ as computing a graph  covered by $\mathbf G.$
Conversely, as a ``translation'' from \cite[Th. 5]{MazurEnum}, 
 every graph covered by $\mathbf G$ 
can be obtained by a run of the algorithm.

\begin{proposition}\label{calcul}
  Let $\mathbf G$ be a labelled graph.
  \begin{enumerate}
  \item For all runs $\rho$ of $\mk$, $\mathbf G$ is a covering of 
    $\gse.$
  \item (completeness) For all $\mathbf H$ such that $\mathbf G$ is a 
    covering  of $\mathbf H$,  there exists a run $\rho$ such that
    ${\mathbf H}\simeq \gse.$ 
  \end{enumerate}
\end{proposition}

\begin{proof}
  \begin{enumerate}
  \item Since $n_\rho$ is
    locally bijective (Theorem~\ref{fundamental}.6), we obtain from Lemma 
    \ref{LB_kcov} that $\mathbf G$ is a covering of $\gse.$
  \item We exhibit a run of $\mk$ having the required property.
    Suppose that we have an enumeration of the vertices of $\mathbf H$. Let
    $\mu$ be the labelling of $G$ obtained by lifting the
    enumeration. There is an execution of Mazurkiewicz' algorithm such
    that each vertex $v$ of $G$ gets $\mu(v)$ as a final
    $n_\rho$-labelling. 

    This is done in the following way. First we
    apply the renaming rule to all vertices in $\mu^{-1}(1)$. This
    is possible because there is no overlapping of balls, since
    $\mathbf G$
    is a covering of $\mathbf H$. Then we apply the diffusion 
    rule as long as we can. After that, we apply the renaming
    rule to $\mu^{-1}(2)$. Because of the diffusion, the number 1 is
    known by all the vertices, so the vertices of $\mu^{-1}(2)$
    get labelled by $2$. And so on, until each vertex $v$ gets
    labelled by $\mu(v)$.
  \end{enumerate}
\end{proof}

  From Proposition~\ref{calcul}.1, we can see a  run of $\mk$ as
computing a covering. Furthermore, if the underlying graph is
covering-minimal, then $\rho(\bG)$ is an isomorphic copy of \bG.
 This copy can be computed
from their mailbox by any vertex, providing a ``map'' -- with numbers
of identification -- of the
underlying network. Thus, on minimal networks, the algorithm of Mazurkiewicz
can actually be seen as a {\em cartography algorithm}. 
\label{reconstruct}

\subsubsection{Interpretation of the Mailboxes at the Step $i.$}
The previous results concern the interpretation of the final
mailboxes. Now, we consider a relabelling chain 
$(\mathbf G_i)_{0\leq i}.$ For a given $i$ and a given vertex $v$
we prove that it is possible to interpret the label of 
$v$ in $\mathbf G_i$ as
a graph quasi-covered by $\mathbf G_i.$
We recall notation.
Let $\mathbf G$ be a labelled graph.
Let $\rho$ be a run of the Mazurkiewicz algorithm and let
$(\mathbf G_i)_{0\leq i}$ be a chain associated to $\rho$ with
$({\mathbf G}_0={\mathbf G}).$
If $v$ is a vertex of $\mathbf G$ then  the label of $v$ at step $i$
is denoted by $(\lambda(v),c_i(v))=(\lambda(v),(n_i(v),N_i(v),M_i(v))).$
Using the interpretation of the previous section by defining
$Strong(M_i(v)),$  this label enables in some cases 
the reconstruction of the graph  $\mathbf H_{M_i(v)}.$ We note
\begin{eqnarray*}
  \mathbf {\mathbf H}_i(v)&=&
  \begin{cases}
 {\mathbf  H}_{M_i(v)} \mbox{if it is defined and } 
(n_i(v),\lambda(v),N_i(v)) \in Strong(M_i(v))\\
  \bot \mbox{ otherwise.}
  \end{cases}
\end{eqnarray*}
We prove that $\mathbf G_i$ is a quasi-covering of $\mathbf H_i(v).$ 
First, we need a definition:
\begin{definition}
Let $(\mathbf G_i)_{0\leq i},$ be a relabelling chain obtained with
the Ma\-z\-u\-r\-k\-i\-e\-wicz algorithm and let $v$ be  
a vertex. We associate
to the vertex $v$ and to the step $i$ the integer $r_{agree}^{(i)}(v)$
being the maximal integer bounded by the diameter of
$\mathbf G$ such that  any vertex $w$ of
$B(v,r_{agree}^{(i)}(v))$ verifies: ${\mathbf H_i(v)}={\mathbf H_i(w)}.$
\end{definition}
Now we can state the main result of this section:
\begin{theorem}\label{agree}
Let $(\mathbf G_i)_{0\leq i},$ be a relabelling chain obtained with
the Maz\-u\-r\-k\-i\-e\-wicz algorithm and let $v$  be a vertex. The graph
$\mathbf G_i$ is a quasi-covering of $\mathbf H_i(v)$ centered on
$v$ of radius $r_{agree}^{(i)}(v).$
\end{theorem}
\begin{proof}
Let $r=r_{agree}^{(i)}(v),$ and  
let $\gamma$ be the partial function which associates to the vertex
$u$ of $B_{\mathbf G_i}(v,r)$ the vertex $n_i(u).$ The aim
of the proof is to verify that $\mathbf G_i$ is a quasi-covering 
via $\gamma$ of $\mathbf H_i(v)$ centered on $v$ of radius $r.$

Using notation of the definition of quasi-coverings,
first we define the covering $\mathbf K.$ Let 
$\mathbf B$ be an isomorphic copy of 
$B_{\mathbf G_i}(v,r).$
The graph $\mathbf K$  is obtained by adding to $\mathbf B$
infinite trees defined as follows.

Let $\mathbf U$ be the universal covering of $\mathbf H_i(v).$
Let $x$ be a vertex of  $\mathbf H_i(v),$   let $S$ verifying
$S\subseteq N_{G_i(v)}(x),$ we define  ${\mathbf U}(x,S)$ as
the subtree of $\mathbf U$ obtained by considering walks 
rooted in $x$ such that the first step is of the form $\{x,s\}$
with $s\in S.$

For each vertex $w$ such that $d(v,w)=r,$ we define ${\mathbf U}_w$
as an isomorphic copy to ${\mathbf U}(\gamma(w),S_w)$ with
$$
S_w=\{s\in N_{G_i(v)}(\gamma(w))\mid \forall y\in 
N_{\mathbf G_i}(w)\cap B_{\mathbf G_i}(v,r)\,\,\,\, \gamma(y)\neq s\}.
$$
The copies are disjoint, \ie, if $w\neq w'$ then 
$V({\mathbf U}_w) \cap V({\mathbf U}_{w'})=\emptyset.$

\newcommand{\rconf}{{\ensuremath{r}}\xspace}
  \begin{figure}[htbp]
    \begin{center}
      \input{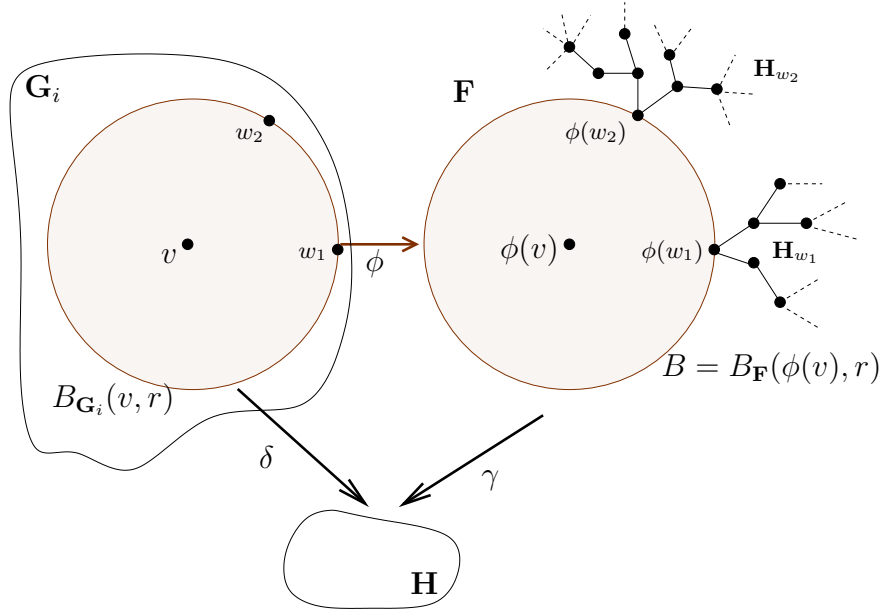}
      \caption{Construction of the associated covering \bF.}
      
    \end{center}
  \end{figure}

For each vertex $w$ such that $d(v,w)=r,$ we add $\mathbf U_w$ 
to $\mathbf B$ by identifying the copy of $w$ and 
the root of $\mathbf U_w.$  Let $\mathbf K$ be the graph we have built
by this way. 

The isomorphism  $\varphi$ is the canonical bijection between
$B_{\mathbf G_i}(v,r)$ and $\mathbf B.$ 

We define the morphism $\delta$ from $\mathbf K$ to $\mathbf H_i(v)$
by:
\begin{itemize}
\item if $u\in B$ then $\delta(u)=\gamma(\varphi^{-1}(u)),$ and
\item if $u\in U_w$ and $t$ is the end-vertex  of the path in 
$\mathbf H_i(v)$ corresponding to $u$ then $\delta(u)=t.$
\end{itemize}

First, we verify that $\delta$ is a morphism. There are three cases:
\begin{itemize}
\item if $u$ and $v$ are adjacent in $\mathbf U_w$ then by construction
$\delta(u)$ and $\delta(v)$ are adjacent,
\item if $u$ and $v$ are adjacent in $\mathbf B$ then they are
adjacent in $\mathbf H_{M_i(u)}=\mathbf H_{M_i(v)}$ 
thus $\delta(u)$ and $\delta(v)$ are
adjacent in $\mathbf H_i(v),$
\item if $u$ and $v$ are adjacent and $u$ belongs to $\mathbf B$
and $v$ belongs to $\mathbf U_w$ for some $w,$ by construction of
$\mathbf U_w$ the vertices $\delta(u)$ and $\delta(v)$ are adjacent in
$\mathbf H_i(v).$
\end{itemize}

By construction, $\delta$ is surjective.
To achieve the proof we verify that 
for all vertices $u$ the restriction of $\delta$ to
$N_{\mathbf K}(u)$ is a bijection onto $N_{\mathbf H_i(v)}.$ 

Once more there are three cases.
\begin{itemize}
\item If $u\in \mathbf H_w$ for some vertex $w$ and if $u$ is
not the root of $\mathbf H_w$ then, by definition of the universal
covering, the restriction of $\delta$ to
$N_{\mathbf K}(u)$ is a bijection onto $N_{\mathbf H_i(v)}.$
\item If $u\in B_{\mathbf B}(\varphi(v),r-1).$ We prove that
the restriction of $\gamma$ to $N_{\mathbf G_i}(\varphi^{-1}(u))$
is a bijection onto $N_{\mathbf G_i(v)}(\gamma(u)).$
By definition of $\gamma$ the restriction is surjective; furthermore
two vertices in the same ball of radius $1$ have different numbers
(it is a direct consequence of the Mazurkiewicz algorithm) thus
the restriction is also injective.
\item If $d(u,\varphi(v))=r$ then using the same argument that for
the previous item combined with the definition of the universal
covering we obtain the result.
\end{itemize}
\end{proof}
\begin{remark}
The previous result remains true for any radius bounded
by $r_{agree}^{(i)}(v).$
\end{remark}

\subsection{An Algorithm to Detect Stable Properties}
In this section we describe in our framework the algorithm by
Szymanski, Shy and Prywes  (the SSP algorithm for short) \cite{SSP}.
\subsubsection{The SSP Algorithm}
We consider a distributed algorithm which terminates when
all processes reach their local termination conditions.
Each process is able to determine only its own termination condition.
The SSP algorithm detects an instant in which the entire
computation is achieved.

Let $G$ be a graph, to each node $v$ is associated a predicate $P(v)$
and an integer $a(v).$ Initially $P(v)$ is false and $a(v)$ is equal
to $-1.$
The relabelling rules are the following, let $v_0$ be a vertex and let $\{v_1,...,v_d\}$ the set of vertices 
adjacent to $v_0.$ 
If $P(v_0)=false$ then $a(v_0)=-1;$
if $P(v_0)=true$ then $a(v_0)= 1 + Min\{a(v_k) \mid 0\leq k \leq d\}.$

\subsubsection{A Generalization of the SSP Algorithm}

We present here a generalization of the hypothesis under which the SSP
rules are run. For every  vertex $v$, the value of $P(v)$ is no more a
boolean and can have any value. Hence we will talk of the {\em
  valuation} $P$. Moreover, we {\em do not} require each
process to determine when it reachs its own termination
condition. Moreover the valuation $P$ must verify 
the following property: for any  $\alpha$, if $P(v)$ has the value
$\alpha$ and changes to $\alpha'\neq\alpha$ then it cannot be equal
to $\alpha$ at an other time. In other words, under this hypothesis,
the function is constant between two moments where it has the same
value. We say that the valuation $P$ is {\em value-convex}.

We extend the SSP rules and we shall denote by GSSP this generalisation. In
GSSP, the counter of  $v$ is incremented only if $P$ is constant on
the ball $B(v)$. As previously, every underlying rule that computes in
particular $P(v),$ has to be modified in order to eventually
reinitialize the counter. Initially $a(v)=-1$ for all vertices.
The GSSP rule modifies the counter $a$.

\label{grs:gssp}
\begin{rrule}{\strut}
  \ritem[\textsc{Rule\_forGSSP}]{Modified rule for GSSP}{
    \item \dots
    \item unchanged
    \item \dots %
      }{
    \item \dots
    \item unchanged
    \item \dots
    \item  For every vertex $v$ of $B(v_0)$, \\
      if $P'(v) \neq P(v)$ then
      \begin{itemize}
      \item       $a'(v):=-1.$ 
      \end{itemize}
      otherwise
      \begin{itemize}
      \item       $a'(v):=a(v).$
      \end{itemize}
      }
    \ritem[GSSP]{GSSP rule}{
    \item For all $v\in B(v_0), \; P(v) = P(v_0)$,
      }{
    \item  $a'(v_0) := 1 +  \min\{a(v) \mid v\in B(v_0)\}.$
      }
\end{rrule}

We shall now use the following notation. Let  $(\mathbf G_i)_{0\leq i
  }$ be a relabelling chain associated to the algorithm GSSP. We denote 
by    $a_i(v)$
(resp. $P_i(v)$) the value of the counter (resp. of the function)
  associated to the vertex $v$ of $\mathbf G_i.$ According to the
  definition of the GSSP rule, we remark that for every vertex $v$,
  $a(v)$ can be increased, at each step, by 1 at most and that if
  $a(v)$ increases from $h$ to $h+1$, that means that at the previous
  step, all the neighbours $w$ of $v$ were such that  $a(w)\geq h$ and
   $P(w)=P(v)$.  The following lemma is the iterated version of this 
   remark.
\begin{lemma}\label{iterated_a}
  For all $j$, for all $v$, for all $w\in B(v,a_j(v))$, there exists
  an integer $i\leq j$ such that
  \begin{eqnarray*}
    a_i(w)& \geq& a_j(v) - d(v,w),\\
    P_i(w)&=&P_j(v).
  \end{eqnarray*}
\end{lemma}
\begin{proof}
  The proof is done by induction upon the radius $k\in[0,a_j(v)]$ of
  the ball. For  $k=0$, the result is true trivially.

  Suppose that the result is true for all vertices in the ball
  $B(v,k)$, $k\leq a_j(v)-1$. Now, we consider a vertex  $w$ at
  distance $k+1$ of  $v$. The vertex $w$ has a neighbour $u$ such that
  $d(v,u)=k$. By induction hypothesis, there exists $i_u\leq j$ 
  such that 
  $a_{i_u}(u)\geq a_j(v)-k$ and $P_{i_u}(u)=P_j(v)$.
 
  Let  $i$ be a step, in the steps preceding $i_u$, where the counter $u$
  reached $a_{i_u}(u)$ with $P_i(u)=P_{i_u}(u)$. This step exists for
  the counter increases of at most 1 at a time, and each time that
  $P(u)$ is modified, the counter $a(u)$ is reinitialized to  $-1$
  (modified rules for GSSP).

  Moreover, according to GSSP rule, we have necessarily,  $a_i(w)\geq
  a_{i_u}(u) -1$ and  $P_i(w)=P_i(u)$. Consequently $a_i(w) \geq
  a_j(v) -k -1$ and $P_i(w)=P_j(v)$. The result is true for  $w$ and
  so for every vertex at distance $k+1$.
\end{proof}

\newcommand{\aaa}{\ensuremath{\lfloor\frac{a_j(v)}{3}\rfloor}}
In particular, this proves that at any moment $j$, for all $v$, for
all  $w\in B(v,a_j(v))$, there exists a moment $i_w$ in the past such
that $P_{i_w}(w)=P_j(v)$. We now prove that, for all vertices  $w$ in
the ball \aaa, we can choose  the same $i_w$. This is a fundamental
property of GSSP algorithm.
\begin{lemma}[GSSP]\label{techlemmaGSSP}
   Consider an execution of the GSSP algorithm under the hypothesis
   that the function $P$ is value-convex. For all $j$, for all
   $v$, there exists $i\leq j$ such that for all $w\in B(v,\aaa)$,
   $P_i(w)=P_j(v)$.
\end{lemma}
\begin{proof}\strut\par
  \begin{figure}[htbp]
    \begin{center}
      \input{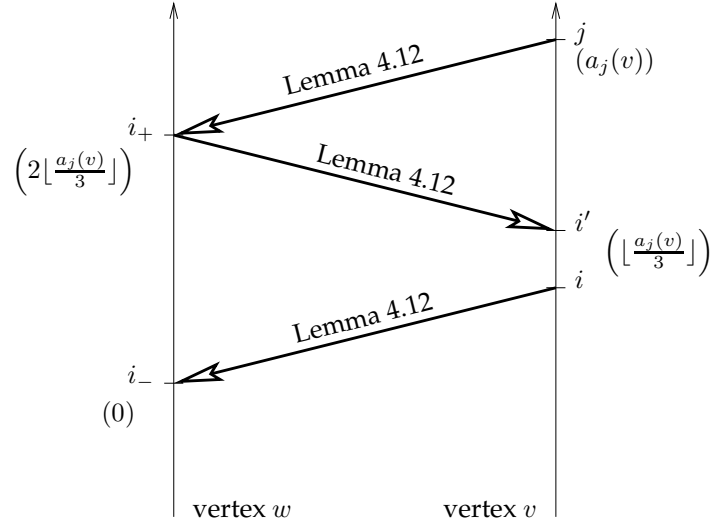}
      \caption{Proof scheme of Lemma \ref{techlemmaGSSP}: vertical
        axes denote time, the value between brackets are lower
        bounds for  the counter $a$.} 
      \label{fig:GSSP}
    \end{center}
  \end{figure}

    Let $i$ be the first step where $a_i(v)=\aaa$ and
    $P_i(v)=P_j(v)$. Let $w\in B(v,\aaa)$, and denote  $i_{+}$ a step,
    which existence is given by Lemma \ref{iterated_a}, such that
    $a_{i_{+}}(w)\geq a_j{v} -d(v,w)\geq 2\aaa$  and
    $P_{i_+}(w)=P_j(v)$. Now, let's apply 
    Lemma \ref{iterated_a} with
    center  $w$ at step $i_{+}$. We obtain  $i'\leq i_{+}$ such that  
    $a_{i'}(v)\geq  a_{i_{+}}(w)-d(w,v)\geq \aaa = a_i(v)$ and
    $P_{i'}(v)=P_{i_+}(w)=P_j(v)$. By minimality of  $i$, $i\leq i'$
    and finally $i \leq i_{+}$.

    Now, we apply another time Lemma \ref{iterated_a},
    with center $v$, at step $i$. We obtain then $i_{-}\leq
    i$, such that $a_{i_{-}}(w)\geq
    a_i(v)-d(v,w)\geq 0$ and  $P_{i_-}(w)= P_i(v)$.

    To conclude, we obtain two steps  $i_{-}$ and  $i_{+}$ such that
     $P_{i_{-}}(w)=P_{i_{+}}(w)=P_j(v)$, and $i_{-}\leq i\leq
    i_{+}$. As $P$ is value-convex, we get
    $P_i(w)=P_j(v)$.
\end{proof}

\begin{remark}
  One third of the counter is an optimal radius of stability. It is
  possible to construct examples where the function $P$ is not
  necessarily constant on the ball of center $v$ and of radius $\aaa+1$.
\end{remark}

   \macro{\algo}{\mathcal A}

In these settings, even if the valuation stabilizes, GSSP is always
guaranted to not terminate. In order 
to have noetherian relabellings systems, we define, given a
relabelling system \grs to which GSSP is applied, $\algo(\grs,P,\varphi)$
to be the relabelling system based upon \grs and GSSP with valuation $P$ and
adding $\varphi$ to the preconditions of the GSSP Rule. Now termination
is closely related to the properties of $P$ and $\varphi$. We define a
property, that is only sufficient, for termination. 

\begin{definition}
  The pair $(P,\varphi)$ is {\em uniform} if for any run of $\algo(P,\varphi)$,
  there exist time $i$ and  $r_0\in \N$ such that for all vertex $v$,
  for all $j\geq i$, we have $\aaa = r_0 \;\Leftrightarrow\;    \neg \varphi_j(v)$. 
\end{definition}

Obviously, a uniform pair implies that $\algo(P,\varphi)$ is noetherian.
The termination of the increase of $\aaa$ at a node $v$ that
first stops, does not prevent the counter at the other  nodes to reach
this particular value.

\subsection{Mazurkiewicz Algorithm + GSSP Algorithm =
  Universal Local Computation}

The main idea in this section is to use the GSSP algorithm in order to
compute, in each node, the radius of stability of \mk. In other words,
each node $u$ will know how far other nodes agree with its
reconstructed graph  $\bH_{M(u)}$. Let  ${\mathbf G }=(G,\lambda)$ be
a labelled graph, let $(\mathbf G_i)_{0\leq i}$ be a relabelling chain
associated to a run of Mazurkiewicz' Algorithm on the graph $\mathbf
G.$ The vertex $v$ of $\mathbf G_i$ is associated to the label
$(\lambda(v),(n_i(v),N_i(v),M_i(v))).$ Using the interpretation of
section \ref{reconstruct}, this labelling enables to construct a
graph that is an asynchronous network snapshot, a would-be cartography
of the network.

We now assume the main relabelling system to be \mk, the valuation to
be $\bH$. We will have to work a bit on $\varphi$ in order to get a
noetherian system. We denote by
$\algo_0$ the system $\algo(\mk,\bH,false)$. 
The output of $\algo_0$ on the node $v$ is $< \mathbf H_i(v),a_i(v)
>.$

\newcommand{\carto}{\ensuremath{\mathcal{ AS}}\xspace}

Looking only at the labels, we have, from
Theorem~\ref{fondM}, that \bH is value-convex. Then,
from Lemma~\ref{techlemmaGSSP} and Theorem~\ref{agree}, we get the
main property of the computation of $\algo_0$:  
\begin{theorem}[quasi-covering progression] \label{qcov_prog} At
  all step $j$, for all vertex $v$, the output of $\algo_0$ on $v$
  is a couple $< \mathbf H_j(v), a_j(v) > $ such that if
  $\mathbf H_j \neq \bot$, %
  then there exists a previous step $i<j$, such that $\mathbf G_i$
  is a quasi-covering of $\mathbf H_j(v)$ of center $v$
  and of radius $\lfloor \frac{a_j(v)}{3}\rfloor$.
\end{theorem}

And as the underlying Mazurkiewicz Algorithm is always terminating, we
have that the value of \bH will stabilize with $a$ going to the
infinite. 

\macro{\rtrust}{r^{\mathrm{t}}}

Finally, and considering the previous theorem, we note
$\rtrust(v)=\lfloor \frac{a_j(v)}{3}\rfloor$, the {\em radius of trust}
for the algorithm $\algo_0$ at node $v$. In the following section, we
show how to get a noetherian relabelling system from $\algo_0$.

\subsection{Computing Further Informations}

   \macro{\grsdxi}{\grs^\dxi}

The following sections show that \mk, Mazurkiewicz' algorithm, is
actually computing the maximal information we can get
distributively. In this subsection, we precise how to extract more
information if we have more structural knowledge. In the following, we
suppose we know a recursive graph family \gfam to which the underlying
network belongs.

Despite the non-recursivity of \covgfam (cf
Remark~\ref{covgfamnotrecursive}), we explain how to use the
 algorithm $\algo$ and Theorem \ref{qcov_prog} to
determine when \bH is in \covgfam. In the remaining of this part, we
suppose $\algo_0$ is running and all considered \bH and $r$ are outputs
of $\algo_0$.

Given these outputs, we describe a (sequential) computation that is
done by all the nodes. 

   \macro{\dxi}{{\chi_\covgfam}}

   \dontprintsemicolon %
   \newcommand{\Com}[1]{{\it \hspace{1cm}/* #1 */}}
   \begin{algorithm}[H] \label{compf}
     \caption{ \dxi: Knowing if \bH is in \covgfam.}
     \KwData{%
       a graph ${\mathbf H}\in\lgraph$,\\
       $r\in\N$.%
}
     \KwResult{$\bot$ or $Yes$ }
     \Repeat{ \bK is a quasi-covering of center $u$ and radius $r$ of \bH
       \Com{Loop ends by Theorem \ref{qcov_prog}
     }}{  
       $\bK\in\gfam$\Com{Enumerate (always in the same order) all the
         graphs of \gfam by order of increasing diameter} \;
     }
     \eIf{\bK is a quasi-covering of radius $r$ of \bH for any vertex
       of \bH and $r > \Delta(\bH)$}{%
       \KwOut{Yes}%
     }{\KwOut{$\bot$} }
   \end{algorithm}

   We note $\dxi(\bH,u,r)$ the result of this (semi-)algorithm. If one
   of the input is clear from the context, it is omitted.

   \begin{lemma}\label{lemmadxi}
     Suppose a graph family \gfam is given.
     For any graph \bH, any vertex $u$, any $r\in\N$, if
     $\dxi(\bH,u,r) \mbox{ terminates and outputs } Yes$, then
     $\bH\in\covgfam$. 
   \end{lemma}
   \begin{proof}
     Denote $\bK$ the
     quasi-covering that ends the loop of \dxi for input
     $\gfam,\bH,u,r$.
     As $r\geq \Delta(\bK)+1$, \bK is a non-strict quasi-covering.
     $\bK$ is then a covering
     of \bH.
     
     The graph $\bK$ being in \gfam, we have that $\bH\in\covgfam$.
   \end{proof}
   
   \begin{lemma}\label{memer}
     Let $\bH\in\covgfam$, for all vertices $u,v$ in \bH, for all
     $r$, $$\dxi(u,r)=Yes \mbox{ iff } \dxi(v,r)=Yes.$$
   \end{lemma}

   \begin{proof}
     Suppose $\dxi(u,r)=Yes$. Denote $\bK$ the
     quasi-covering that ends the loop of \dxi for input $u$. 

     By condition for output $Yes$, we have that $\bK$ is also a
     quasi-covering of center $v$ and radius $r$ of \bH, hence  $
     \dxi(v,r)=Yes.$  
   \end{proof}

   We define 
   $\varphi_I$ to be 
   $\dxi(\bH,n(v_0),r^{\mathrm{t}}(v_0)) \neq Yes$.
   
   \renewcommand{\carto}{\ensuremath{\mathrm{Carto}}\xspace}

   We note \carto the system $\algo(\mk,\bH,\varphi_I)$. This will give a 
   noetherian system as demonstrated below. The
   output of \carto is $( \bH, \rtrust )$.

   \begin{theorem}[Asynchronous snapshot] \label{prebH}\label{dxi}
     With inputs $\bH$ and $r^{\mathrm{t}}$ computed by the
     relabelling system \carto, we have the following properties.
     \begin{theoenum}
     \item The semi-algorithm \dxi with inputs \bH and $\rtrust$ terminates.
     \item \label{bHok}  At any time, if $\dxi(\bH,\rtrust)=Yes$, then
       $\bH\in\covgfam$.  
     \item If \bH is defined, then there exists a previous step $i$, such
       that $\mathbf G_i$   
       is a quasi-covering of $\bH$ of center $v$ and of radius
       $\rtrust$. 
     \end{theoenum}
   \end{theorem}
   \begin{proof}
     The first property is given by Theorem~\ref{qcov_prog}.
     The second one is then by Lemma~\ref{lemmadxi}.
     
     As every run of \carto is a prefix of a run of $\algo_0$, we get
     the final property by \ref{qcov_prog}.
   \end{proof}

   \begin{theorem}\label{dxinoeth}\label{algo0noeth} 
     The system \carto is noetherian.
   \end{theorem}
   \begin{proof}
     We  show that $(\bH,\varphi_I)$ is uniform.
     Until \mk terminates, the modified GSSP part of the system has no
     significant consequences (the computations of $r^{\mathrm{t}}$
     is resetted whenever a rule of \mk is applied). When
     \mk is finished,  $r^{\mathrm{t}}$ starts to increase. It will
     increase until the computation of $\dxi(\bH,r^{\mathrm{t}})$
     outputs $Yes$ on some node $v$.  

     At this moment, and at this node $v$, $\varphi_I$ is no more
     true. As we are working with the final labelling of \mk, \bH has 
     the same value on all nodes, hence from Lemma \ref{memer}, the
     computation of \dxi will output $Yes$ for the same value of
     $r^{\mathrm{t}}$ on every node.

     Then, \carto is noetherian. 
   \end{proof}

   \begin{remark}
     As a small optimisation for \dxi, it shall be noticed that it is
     not to be run if \rtrust is smaller than the diameter of \bH.
   \end{remark}

   Without further informations about \gfam, it seems difficult to deduce
   anything more.

\section{Termination Detections}
First, we recall from the previous section:
let $\grs$ be a locally generated relabelling relation,
let $\mathbf G$ a labelled graph, we say that $\mathbf G$ is an
irreducible configuration modulo $\grs$ if $\mathbf G$ is a $\grs$-normal
form, \ie,  no further step with  $\grs$ is possible
(${\mathbf G}{\grs} {\mathbf G'}$ holds for no $\mathbf G'$).

Irreducibility with respect to a relabelling relation yields a notion
of implicit termination: the computation has ended - no more
relabelling rule can be applied - but no node is aware of the
termination. On the other hand, one shall ask a node to be 
 aware of the termination of the algorithm. We will see that one
can define various flavour of detection of the termination of a
local computation.
\begin{itemize}

\item Termination of the algorithm but without detection: {\em implicit termination}
\item The nodes end in a {\em terminal state}, that is a state in which the
  node knows it will stay forever (whatever is happening elsewhere in the
  network): {\em detection of the local termination}
\item 
The nodes know when all other nodes have computed their final output
value.  This is the {\em observed termination detection} as when termination
is detected, some observation computations are not necessarily
terminated. Due to the asynchronous aspect of local computations,
there is still some {\em observational} computations that are going
on. 
\item A node enters a special state that indicates that the 
  algorithm has {\em terminated}. This is, obviously, the last step of
  the computation. 
\end{itemize}

The three last cases are explicit terminations. Termination of a
distributed algorithm is usually implicitly assumed to be (one kind
of) explicit.

We will see that these various notions are distinct and form a strict
hierarchy.  First we will give the formal definitions, some examples
and then the characterisations of each termination detection. The
characterisations are complete except for the local termination
detection where we have results only for uniform tasks, that is,
local computations ending in a uniform labelling of the
network.

\macro{\mem}{{\texttt{mem}}}
\macro{\xlabel}{{\texttt{x}}}
\macro{\out}{{\texttt{out}}}
\macro{\lterm}{{\texttt{term}}}

\newcommand{\term}{\textsc{Term}\xspace}

\subsection{Normalisation of the Labellings}

In order to have a unified presentation of the various results, we
restrict ourselves to ``normalised relabelling systems'' w.l.o.g. 

\begin{definition}

  A {\em normalised labelled graph} \bG is a labelled graph whose
  labelling is of the form $(\mem,\out,\lterm).$

  A {\em normalised relabelling system} \grs is a graph
  relabelling system on normalised graphs where 
  \begin{itemize}
  \item \mem can be used in preconditions and relabelled,
  \item \out is only relabelled,
  \item \lterm is only relabelled and has a value in $\{\bot,\term\}$. 
  \end{itemize}

  We also use the following convention: if the initial labelled
  graph is $\bG=(G,in)$ then it is implicitly extended to the
  normalised labelling $\left(G,(in,\bot,\bot)\right)$. The initial
  value of \mem is therefore given by $in$.
\end{definition}

Now, we make the following assumptions. All graphs are labelled graphs
and are all considered to be  normalised. All relabelling relations
are relabelling relations of normalised labelled graphs.

We also use the following notations. Let \bG and \bG' be some given
normalised graphs then, for any vertex $u\in\bG$ (resp. $\in\bG'$), for 
any $\xlabel\in\{\mem,\out,\lterm\}$, $\xlabel(u)$
(resp. $\xlabel'(u)$) is the \xlabel component of $u$ in \bG
(resp. \bG'). 

The graph $\xlabel_\bG$ is the labelled graph obtained from \bG by keeping
only the $\xlabel$ component.
\medskip

This presentation will find its justifications with the following
definitions.

\newcommand{\distask}{task\xspace}
\newcommand{\distasks}{tasks\xspace}
\newcommand{\Distasks}{Tasks\xspace}

\subsection{Termination Detection of Relabelling Systems and of
  \Distasks}

In this section, we give the formal definitions of the termination detection
for graph relabelling systems. Then we define what is a \distask,
informally it is a relation between the set of inputs 
and the set of ``legal'' outputs. 

Let \gfam be a given graph family, and \grs a graph relabelling
system. We denote by $Im_\grs(\gfam)$ the set 
$$\{\bG'\in\lgraph \mid \exists \bG\in\gfam, \bG\grs^*\bG'\}.$$

\subsubsection{Implicit Termination} There is no detection
mechanism. Hence \lterm is not used.

\begin{definition}
  A graph relabelling system \grs has an {\em implicit termination}
  on \gfam if 

  \begin{theoenum}
  \item \grs is noetherian on \gfam,
  \item the \lterm components of any graph in
    $Im_\grs(\gfam)$ are all equal to $\bot$.
  \end{theoenum}
\end{definition}

If the underlying local computation is aimed at the computation of
a special value, we will, in order to distinguish this value 
from the intermediate computed values, only look the special purpose
component \out. As there is no detection of termination, this label
is written all over the computation. It becomes significant only when
the graph is irreducible, but no node knows when this happens.

\begin{remark}
  Such a definition is also relevant for finite
  self stabilising
  algorithms \cite{dolev}. %
  Indeed, one can see implicit termination
  as a stabilisation. Furthermore, Mazurkiewicz' algorithm has been
  shown to be selfstabilizing \cite{SelfstabEnum}. 
\end{remark}

\subsubsection{Local Termination Detection}

We will now ask the \out label to remain unchanged once \lterm is set
to \term. 

\begin{definition}
  A graph relabelling system has a {\em local termination detection}
  (LTD)  on \gfam if 
  \begin{theoenum}
  \item \grs is noetherian on \gfam,
  \item \lterm components of graphs in $Irred_\grs(\gfam)$ are equal to
    \term,
  \item for all graphs $\bG,\,\bG'\in Im_\grs(\gfam)$ such that $\bG\grs^*\bG'$,
    for every vertex $u$ such that $\lterm(u)=\term$,  then 
    \begin{eqnarray*}
      \out(u)&=&\out'(u),\\
      \lterm(u)&=&\lterm'(u)=\term.
    \end{eqnarray*} 
  \end{theoenum}
\end{definition}

 \begin{remark}
   It shall be noted that this definition does not
   formally prevent a node in a terminal state to act as a gateway by
   maintaining connectivity of the active parts of the network. Note
   that the \mem component can also be left unchanged with rules that
   relabel only neighbours. We do not discuss here if, in some sense, a node 
   acting only as gateway is really ``terminated''. Furthermore, we will
   only give results for uniform tasks (to be defined later)  where
   these distinctions actually give equivalent definitions. 
 \end{remark}

\subsubsection{Observed Termination Detection}

In this section, we require that once \term appears, all \out values
have to remain unchanged until the end of the relabellings. 

\begin{definition}
  A graph relabelling relation \grs has an {\em observed termination
    detection} (OTD) on \gfam if 
   \begin{theoenum}
   \item \grs is noetherian on \gfam,
   \item \lterm components of graphs in $Irred_\grs(\gfam)$ are equal to
    \term,
   \item for all graphs $\bG,\,\bG'\in Im_\grs(\gfam)$ such that
     $\bG\grs^*\bG'$, for all vertex $u$ such that $\lterm(u)=\term$,  then
     \begin{itemize}
     \item $\lterm'(u)=\term$,
     \item for all vertex $v\in\bG$, $\out(v)=\out'(v).$ 
     \end{itemize}
   \end{theoenum}
\end{definition}

In other words, every node can know when every output value is final. The
point is that, in this definition, we ask the network to detect the
termination of the computation (in the sense of the \out value that is
computed), but not to detect the termination of that detection. In the
following, we usually have one vertex that detects that the \out values
are final and then it performs a broadcast of \term. This is actually the
termination of this broadcast that we do not ask to be detected. In
some sense, this broadcast is performed by an ``observer algorithm''
whose termination we do not consider.

\begin{remark}
  Up to a broadcast, this definition is equivalent to a ``weaker'' one
  where it is  asked that only at least one vertex of irreducible
  graphs has a \lterm label set to \term. 
\end{remark}

\subsubsection{Global Termination Detection}
\macro{\tasks}{\mathfrak T}
\macro{\task}{T}

There is a node that performs explicitly the last relabelling rule.

\begin{definition}
  A graph relabelling system \grs has a {\em global termination
  detection} (GTD) on \gfam if 
   \begin{theoenum}
   \item \grs is noetherian on \gfam,
   \item for all graphs $\bG\in Im_\grs(\gfam)$, there exists a vertex
     $u$ such that $\lterm(u)=\term$ if and only if $\bG$ is
     in $Irred_\grs(\gfam)$.%
  \end{theoenum}
\end{definition}

\subsubsection{Termination Detection of \Distasks} We now
  define \distasks by a specification and a domain. The
  specification is the general description of what we want to do. The
  domain is the set of labelled graphs where the local computation has to
  compute the correct outputs with respect to the specification. 

  First we recall some basic definitions about relations.
  \begin{definition}
    A relation $R$  is left-total on a set $X$ if for every $x\in X$,
    there exists $y$, such that $xRy$.
  \end{definition}

  \begin{definition}
    Let a relation $R$ on the set \lgraph of labelled graphs. Let
    $X\subseteq \lgraph$. The {\em restriction} of $R$ to $X$ is the
    relation $R_{\mid X}=R\cap X\times\lgraph.$
  \end{definition}
  
  \begin{definition}
    A  {\em \distask} \task is a couple $(\gfam,S)$ where
    $S$ is a (not necessarily locally generated) relabelling
    relation and  \gfam is a recursive labelled graph family, such
    that $S_{\mid \gfam}$ is left-total on \gfam.

    The set \gfam is the {\em domain} of the task, $S$ is the {\em
    specification} of the task.  

  \end{definition}

    Note that, in general, a specification $S$ is not particularly related
  to a given graph family. However, the computability of a
  \distask does depend on the domain. See the Election Problem in
  Sect.~\ref{election}.  

   A specification can be a decision task such as recognition of
  property of the underlying graph, or consensus
  problems, or the problem of election of a node (see Section \ref{election}),
  a task depending on the level of structural knowledge we have, the
  computation of a spanning tree, a $d-$colouration of a graph, etc....

  \begin{remark}
    It shall be emphasised that we do not explicitly deal with
    structural knowledge as a parameter for the algorithm. This is
    {\em exactly the same algorithm that is applied on all 
      the labelled graphs}. If there is any parameter to use, this has to 
    be done in the description of the domain and, maybe, encoded in the
    initial label.

    Our definition is aimed at emphasising the difference between the
    problem - that is the same for any network - and the set of
    networks on which we want to solve it ( if it is solvable at all
    ) with a unique  algorithm. E.g., in Section
    \ref{election}, we show that, for any minimal graph, there is
    an Election algorithm but there is no algorithm that solve
    Election for all minimal graphs. 
  \end{remark}

  Keeping in mind the previous remark about structural knowledge, any
  intuitive \distask can be encoded by this way. 
  \begin{example}
    We describe the specification of the $d-$colouration problem:
    \begin{eqnarray*}
      colo_d &=& \{(\bG,(\bG,\lambda))\;\mid\; \card(\lambda(\bG))\leq d,
      \mbox{ and } \\
    &&\forall (u,v)\in E(\bG), \lambda(u)\neq\lambda(v)\}.
    \end{eqnarray*}

    A solution to the task $(R,colo_3)$ is presented in Example \ref{colod}.
  \end{example}

  We now define the computability of a \distask with respect to the
  different flavours of termination.

  \begin{definition}
    A task $(\gfam,S)$ is locally computable with implicit termination
    (resp. LTD, OTD, GTD) if there exists a graph relabelling system
    \grs such that  
    \begin{theoenum}
    \item (termination) \grs has an implicit termination (resp. LTD, OTD,
      GTD) on \gfam,
    \item (correctness) for any graphs $\bG\in\gfam$, 
      $\bG'\in{\mathrm Irred}_\grs(\bG)$ $$ \bG S \out_{\bG'},$$
    \item (completeness) for any graph $\bG\in\gfam$, for any graph
      $\bG'$ such that $ \bG S \bG',$ there exists $\bG''$ such that 
      \begin{eqnarray*}
      \bG''&\in&\mathrm{Irred}_\grs(\bG),\\
      \bG'&=&\out_{\bG''}.
      \end{eqnarray*}
    \end{theoenum}
    
    In this case, we say that the graph relabelling relation \grs { \em
      computes} the task $(\gfam,S)$ with no (resp. local, observed,
    global) termination detection.
  \end{definition}

  \begin{remark}
    The reader should remark that previous definitions (\cite{YKsolvable,BV0})
    are restricted to the {\em correctness} property. This is the 
    first time, to our knowledge, that the {\em completeness} (that
    can be seen as a kind of fairness property over the legal outputs) is
    addressed, thus giving its full meaning to the sentence ``{\em $S$
      is locally computed by \grs on \gfam}''. 

    Moreover, the impossibility results remain true even without the
    completeness condition (see Remark~\ref{aboutcompleteness}).
  \end{remark}

  \begin{remark}
    The terms message termination and process termination have also
    been used to denote implicit and explicit termination
    \cite[introduction for chap. 8]{Tel}.
  \end{remark}

  We denote by ${\tasks}_{\mathrm I}(\gfam)$ (resp. ${\tasks}_{\mathrm
    {LTD}}(\gfam)$, ${\tasks}_{\mathrm  {OTD}}(\gfam)$, ${\tasks}_{\mathrm
    {GTD}}(\gfam)$ )  
  the set of specifications that are locally computable on domain
  \gfam with implicit termination (resp. LTD, OTD, GTD). If \gfam is
  obvious from the context, we will omit it in these notations. 

  From the definitions, we have   
  \begin{proposition}\label{hierarchy} For any labelled graph family \gfam,
    $${\tasks}_{\mathrm{ GTD}}(\gfam) \subset{\tasks}_{\mathrm
      {OTD}}(\gfam)  \subset{\tasks}_{\mathrm{ LTD}}(\gfam) 
    \subset{\tasks}_{\mathrm I}(\gfam).$$ 
  \end{proposition}

\begin{proof} We give the proof, from right to left, as an
  illustration for those definitions. 

  A task $T$ with local termination detection has an implicit
  termination: remove relabelling of \lterm in a relabelling system
  that computes $T$ with LTD. 
  
  Suppose a task $T$ is computable with observed termination
  detection. A relabelling system that computes $T$ with OTD has LTD by
  definition.  
  
  Suppose now that $T$ is computable with global termination
  detection by \grs. An OTD system for $T$ can be obtained by adding
  a \term-broadcast rule to a relabelling system that computes $T$
  with GTD. 
\end{proof}

Before we characterise these different
classes and show that they define a strict hierarchy, we present some
examples.

\subsection{Four Examples about Computing the Size of an  Anonymous Tree}
\label{exemple4}

We illustrate
these various kinds of termination with
the example of the computation of the size of a tree. We give
four  algorithms: with implicit termination, with local termination
detection, with distributed termination detection, with global
termination detection. In all cases, we start with the labels of the nodes
being uniformly set to $(0,\bot,\bot)$.

The first three relabelling systems are variations of the fourth
one. Thus we focus on the last relabelling system,
\textsc{TreeSize\_GTD}. The rules are described in their order of
appearance. 

First we prune the tree starting from the leaves. The size
of the pruned sub-tree is 
computed incrementally. When the last
vertex is pruned, it knows it has the total number of vertices. It
broadcasts this value. 

When the leaves get the broadcast value, they acknowledge it to their
neighbour. Then, the last vertex to get acknowledgements from all its 
neighbours knows this is the end of the local computation. It
shall be noted that this is not necessarily the same pseudo-root
vertex  at each wave.  

We recall that $N(v_0)$ is the set of neighbours of $v_0$ and that,
given a ``meta-rule'', we enable only rules that modify at least one 
label. 
Proofs are left as exercises.

\begin{example}\strut
  \begin{rrule}{\textsc{TreeSize\_I}} 
    \ritem{Pruning}{
    \item $\mem(v_0)=0$,
    \item $\exists! v\in N(v_0), \mem(v)=0$ or $\forall v\in N(v_0),
      \mem(v)\neq0$. 
    }{
    \item $\mem'(v_0) = 1 + \sum\nolimits_{v\in N(v_0)} \mem(v)$,
    \item $\out'(v_0) = \mem'(v_0)$.
    }
    \ritem{Fast Broadcast}{
    \item $\forall v\in N(v_0), \mem(v)\neq0$,%
    }{
    \item $\out'(v_0)=\max\nolimits_{v\in N(v_0)}\{\out(v)\}$. 
    }
  \end{rrule}
\end{example}

\begin{example} \strut
 \begin{rrule}{\textsc{TreeSize\_LTD}}
    \ritem{Pruning}{
    \item $\mem(v_0) = 0$,
    \item $\exists! v\in N(v_0), \mem(v)=0$.
    }{
    \item $\mem'(v_0) = 1 + \sum\nolimits_{v\in N(v_0)} \mem(v)$.
    }
    \ritem{Tree Size is Computed}{
    \item $\mem(v_0) = 0$,
    \item $\forall v\in N(v_0), \mem(v) \neq 0$.
    }{
    \item $\out'(v_0) = 1 + \sum\nolimits_{v\in N(v_0)} \mem(v)$,
    \item $\mem'(v_0) = \textsc{Size}$,
    \item $\lterm'(v_0) = \term$.
   }
    \ritem{Broadcast Size}{
    \item $\exists v\in N(v_0), \mem(v) = \textsc{Size}$.
    }{
    \item $\mem'(v_0) = \textsc{Size}$.
    \item $\out'(v_0)=\max\nolimits_{v\in
        N(v_0)}\{\out(v)\}$.
    \item $\lterm'(v_0) = \term$.
    }
  \end{rrule}  
\end{example}

\begin{example} \strut
 \begin{rrule}{\textsc{TreeSize\_OTD}}
    \ritem{Pruning}{
    \item $\mem(v_0)=0$,
    \item $\exists! v\in N(v_0), \mem(v)=0$.
    }{
    \item $\mem'(v_0) = 1 + \sum\nolimits_{v\in N(v_0)} \mem(v)$. 
    }
    \ritem{Tree Size is Computed}{
    \item $\mem(v_0)=0$,
    \item $\forall v\in N(v_0), \mem(v) \neq 0$.
    }{
    \item $\out'(v_0) = 1 + \sum\nolimits_{v\in N(v_0)} \mem(v)$,
    \item $\mem'(v_0) = \textsc{Size}$.
    }
     \ritem{Broadcast Size}{
     \item $\exists v\in N(v_0), \mem(v)=\textsc{Size}$.
     }{
     \item $\mem'(v_0) = \textsc{Size}$,
     \item $\out'(v_0)=\max\nolimits_{v\in
         N(v_0)}\{\out(v)\}$.
    }

    \ritem{End of Broadcast}{
    \item $\card\left( N(v_0)\right) \leq 2$,
    \item $\mem(v_0)=\textsc{Size}$.
    }{
    \item $\mem'(v_0) = \textsc{Ack}$.
    }
    \ritem{Acknowledgement}{
    \item $\exists v\in N(v_0), \mem(v)=\textsc{Ack}$.
    }{
    \item $\mem'(v_0) = \textsc{Ack}$. 
    }
    \ritem{Termination Detection}{
    \item $\mem(v_0) \neq\textsc{Ack}$,
    \item $\forall v\in N(v_0), \mem(v) = \textsc{Ack}$.
    }{
    \item $\mem(v_0)=\term$
    \item $\lterm'(v_0)=\term$.
    }
    \ritem{Broadcast Termination}{
    \item $\exists v\in N(v_0), \mem(v) = \term$.
    }{
    \item $\mem'(v_0) = \term$,
    \item $\lterm'(v_0)=\term$.
    }
  \end{rrule}  
\end{example}

\begin{example} \strut
 \begin{rrule}{\textsc{TreeSize\_GTD}}
    \ritem{Pruning}{
    \item $\mem(v_0)=0$,
    \item $\exists! v\in N(v_0), \mem(v)=0$.
    }{
    \item $\mem'(v_0) = 1 + \sum\nolimits_{v\in N(v_0)} \mem(v)$. 
    }
    \ritem{Tree Size is Computed}{
    \item $\mem(v_0)=0$,
    \item $\forall v\in N(v_0), \mem(v) \neq 0$.
    }{
    \item $\out'(v_0) = 1 + \sum\nolimits_{v\in N(v_0)} \mem(v)$,
    \item $\mem'(v_0) = \textsc{Size}$.
    }
     \ritem{Broadcast Size}{
     \item $\exists v\in N(v_0), \mem(v)=\textsc{Size}$.
     }{
    \item $\mem'(v_0) = \textsc{Size}$.
    \item $\out'(v_0)=\max\nolimits_{v\in
        N(v_0)}\{\out(v)\}$.
    }    

    \ritem{End of Broadcast}{
    \item $\card\left( N(v_0)\right) \leq 2$,
    \item $\mem(v_0)=\textsc{Size}$.
    }{
    \item $\mem'(v_0) = \textsc{Ack}$.
    }
    \ritem{Acknowledgement}{
    \item $\exists v\in N(v_0), \mem(v)=\textsc{Ack}$.
    }{
    \item $\mem'(v_0) = \textsc{Ack}$. 
    }
    \ritem{Termination Detection}{
    \item $\mem(v_0) \neq\textsc{Ack}$,
    \item $\forall v\in N(v_0), \mem(v) = \textsc{Ack}$.
    }{
    \item $\lterm'(v_0)=\term$.
    }
  \end{rrule}  
\end{example}

 \subsection{Computing a Spanning Tree}\label{exemple3}\label{subsec:spantree}

 We consider here the problem of building a spanning tree in a graph.
 We assume that there exists a distinguished vertex,
 all vertices are initially in some neutral state (encoded by the label
 $\bot$) except exactly one vertex which is in an active state 
 (encoded by the label $\epsilon$).

 The construction of a spanning tree for a rooted network is among
 the most fundamental tasks to be performed. The spanning
 tree may be used subsequently 
 for performing broadcast and convergecast communications.

 \subsubsection{Local Computation of a Spanning Tree
 With Detection of the Global Termination.}

The main idea is to use Dewey's prefix-based labelling. The father of
the node $v$ is the neighbour labelled by the prefix of $v$. This
encoding is necessary as, here, we restrict to no label on edges or 
ports.  Whenever the covering algorithm is finished, the leaves acknowledge
the termination to their fathers until the root node knows everything
is over.

The labels \mem are  words upon the alphabet \N. We
note $\alpha.\beta$ the concatenation of the words $\alpha$ and
$\beta$. $\epsilon$ denotes the empty word. 
We define the following notations in order to simplify the description
of the rules. Given a  vertex $v_0$, we define $new(v_0)=\{v\in
B(v_0)| \mem(v)=\bot\}$. We also define the set of neighbours labelled
by a prefix of the center's label. Let $children(v_0)=\{v\in
B(v_0)| \mem(v)\in\mem(v_0).\N\}$
Given a set $X$ of nodes, we note $\sigma_X$ an injective function
$X\longrightarrow \N$.

The tree has a distinguished vertex, labelled
$(\epsilon,\bot,\bot)$, all other nodes are labelled
$(\bot,\bot,\bot)$ .

    \begin{rrule}{\textsc{Spanning tree}}
      \ritem{Spanning Vertices}{
      \item $\mem(v_0)\neq\bot$,
      \item $new(v_0)\neq \emptyset.$
        }{
        \item if $v\in new(v_0)$,
        $\mem'(v)=\mem(v_0).\sigma_{new(v_0)}(v).$  
        }
      \ritem{Acknowledgement}{
      \item $\mem(v_0)\in\N^+$,
      \item $\forall v\in child(v_0), \textsc{Ack}$ is suffix of $\mem(v)$.
        }{
      \item $\mem'(v_0) = mem(v_0)|\textsc{Ack}.$
        }
      \ritem{Global Termination Detection}{
      \item $\mem(v_0)=\epsilon$,
      \item $\forall v\in child(v_0), \textsc{Ack}$ is suffix of $\mem(v)$.
        }{
      \item $\lterm'(v_0)=\term.$
        }
    \end{rrule}

As can be seen from the \lterm label, local termination and global
termination are closely related but will differ on the root
node. We can note that the nodes know their final number from the
first application of a rule (the Spanning Vertices rule), but they do not
terminate in order to convergecast the acknowledgement to the root.

\section{Characterisations}

\subsection{Implicit Termination}

\macro{\covS}{\widehat S}
 
We need the following definitions to express the local symmetry of a
task. 

\begin{definition}
  A graph family \gfam is {\em covering-closed} if for any graphs
  \bG,\bH such that \bG is a covering of \bH,
  $\bG\in\gfam\;\Longrightarrow\; \bH\in\gfam$.

  Let $\gamma:\bG\longrightarrow\bH$ be a covering, let $\bH'$ be a
  relabelling of \bH. Then the {\em lifting of $\bH'$ through $\gamma$} is
  the following labelling: $\forall v\in\bG$, the label of $v$ is the
  label of $\gamma(v)$ in $\bH'$. This labelled graph is denoted 
  $\gamma^{-1}(\bH')$. 
\end{definition} 

The following proposition is obvious.
\begin{proposition} 
  Let \gfam be a graph family. Then \covgfam is the smallest graph
  family containing \gfam that is covering-closed.
\end{proposition}

\begin{definition}
  Let \gfam be a covering-closed graph family. A relation $R$ on
  \gfam is {\em lifting-closed} if
  for all graphs \bG and \bH in \gfam, such that \bG is a covering of
  \bH via $\gamma$, for all $\bH'$, $ \bH R \bH'\;\Longrightarrow\; \bG R
  \gamma^{-1}(\bH')$.
\end{definition}

\begin{definition}
 A relabelling relation $S$ is {\em covering-lifting closed on \gfam}
 if there exists a lifting-closed left-total recursive relation
 $\covS$ on \covgfam such that $$\covS_{\mid \gfam} =
 S_{\mid \gfam}.$$ 
\end{definition}

Reminding Remark~\ref{covgfamnotrecursive}, we underline that
${\covgfam}$,   the domain of $\covS$, is not necessarily
recursive. Hence we require only that \covS is recursive with left
input in \covgfam.

\medskip

The necessary condition relies upon Lifting Lemma \ref{lifting}. This a
classical result since the work of Angluin.  The sufficient condition
uses Mazurkiewicz' algorithm.   This result was first proved in a
slightly different context in  \cite{GMMrecognition}.  In
\cite{GMMrecognition}, the algorithm was quite technically involved.
We give here another, maybe simpler, proof using GSSP.  

\macro{\grsi}{\grs^{\mbox{\textsc{i}}}}

\macro{\choosen}{\mbox{\texttt{choice}}}
\newcommand{\CN}{{\em Necessary Condition.}\xspace}
\newcommand{\CS}{{\em Sufficient Condition.}\xspace}
\begin{theorem}\label{caracITD}
  A \distask $(\gfam, S)$ is locally computable with implicit
  termination if and only if it is covering-lifting closed. 
\end{theorem}
\begin{proof}
  \CN
  Let $(\gfam,S)$ be a \distask that is computable with
  implicit termination.
  
  There exists \grs that locally computes $(\gfam,S)$. We define an 
  extension $\covS$ on \covgfam in the following way: given \bH in
  $\covgfam$, we can apply \grs until an irreducible form is
  obtained (this always happens because of
  Lemma~\ref{covgfamnoetherian}).   We take $\bH\covS \bH'$ for any
  irreducible labelled graph \bH' obtained from \bH.

  By construction, \covS is left-total on \covgfam. We now show that
  $\covS$ meets the properties of the covering-lifting closure
  definition.  

  \smallskip
  First, we show that \covS is lifting-closed.   Let  $\mathbf H$ be a
  labelled graph and $\bG\in\gfam$ with 
  $\gamma:\mathbf G\longrightarrow{\mathbf H}$ a covering. 

  Let $\bH'$ such that $\bH \covS \bH'$. By construction,
  $\bH\RR^*\bH'$ and \bH' is 
  $\grs$-irreducible. Hence by the Lifting Lemma \ref{lifting}, we
  have $\bG\RR^*\gamma^{-1}(\bH')$. Furthermore $\gamma^{-1}(\bH')$ is 
  irreducible as $\gamma^{-1}(\bH')$ is. Then $G\covS
  \gamma^{-1}(\bH')$.  

  Finally, we show that the relations \covS and $S$ are equal on \gfam. Let
  $\bG,\bG'$ such that $\bG \covS \bG'$. Since \grs computes $S$, we
  have that $\bG S \bG'$.

  Let $\bG\in\gfam$. As \grs computes locally $S$ on \gfam, for any
  \bG' such that $\bG S \bG'$, there exists, by completeness, an
  execution that leads to an irreducible form equals to \bG'. Hence
  $\bG \covS \bG'$. 

  Given the previous result, we get \covS is recursive when the left
  member is in \covgfam since $S$ is.

  \medskip

  \CS 
  We suppose $(\gfam,S)$ is covering-lifting closed. We will describe
  a graph relabelling system \grsi that computes $(\gfam,S)$. 

  We first describe a ``naive'' approach. This approach describes what
  is essentially at stake here, but, rigorously, it fails for a
  recursivity reason. This approach is, that at any moment, to take \bH
  the computed asynchronous snapshot with \mk, then choose a $\bH'$
  such that $\bH \covS \bH'$
  and lift the \out labels to the vertices of \bG. By covering-lifting
  closure, and by Prop.~\ref{calcul}, at 
  the end of the computation of \mk, it will give a correct final
  labelling. The real problem of this approach is that, in the general
  case, during the computation, it is not possible to know simply when the
  computed \bH is really in \covgfam. Furthermore by Remark
  \ref{covgfamnotrecursive}, even knowing \gfam, it is not computable
  to decide if a given \bH is in \covgfam. 

  However, from Th.~\ref{bHok}, we have a relabelling system \carto
  that outputs when \bH is in \covgfam. \grsi is obtained by
  adding to \carto the following  rules, for any \bH' such that $\bH
  \covS \bH'$: 
  \begin{rrule}{\grsi}
    \ritem[\grsi-\bH']{Pick an Output}{
    \item $\dxi(\bH,n(v_0), r^t(v_0)) = Yes.$
    }{
    \item $\choosen'(v_0)=\bH'$.
    }
    \ritem[\grsi-\mbox{final}]{Unifying Outputs}{
    \item for all $v\in B(v_0)$, $\choosen(v_0)\preceq\choosen(v).$
    }{
    \item for all $v\in B(v_0)$, $\choosen'(v)=\choosen(v_0).$
    \item for all $v\in B(v_0)$, $\out'(v)=\out_{\choosen(v_0)}\left(n(v)\right).$
    }
  \end{rrule} 

  The final rule ensures that the same \bH' is used all over the
  graph by taking the smallest chosen one. 

  By Prop.~\ref{prebH}, \grsi is noetherian and the label \out is
  ultimately computed. By covering-lifting closure, the final \out
  labelling is correct. Moreover, by Proposition~\ref{calcul}
  (completeness), we get the completeness condition about $S$.
\end{proof}

\begin{remark}\label{aboutcompleteness}
  If we drop the completeness property from the requirement, the proof
  shows that it is only necessary and sufficient to have $\covS_{\mid \gfam}
  \subseteq S_{\mid \gfam}.$
\end{remark}

\begin{remark}
   If it is easy (read recursive) to check whether a given graph is in
  \covgfam\ -- for example if $\gfam$ is covering-closed\  -- the
  algorithm above is very much simplified because the main difficulty
  is to know when a ``Pick an output'' rule can be applied. This
  reveals to be actually the case for almost all  practical cases.  
\end{remark}

\subsection{Local Termination Detection of Uniform Tasks}

The results of this part comes from \cite{GMstructToCS}. They stand only
for uniform tasks, that is, for tasks with a uniform \out label. We
adapt the definitions to the context of this paper and we give the
main result.  
 The complete proofs (that are basically the same up to the notations) and some
applications, in particular about the problem of deducing by
local computations a structural information from another one,
are given in \cite{GMstructToCS}.

\begin{definition}
  A \distask is {\em uniform} if for every $\bG\in\gfam$,
  every \bG' such that $\bG S \bG'$, every vertices $u,v\in\bG$,
  $\out_{\bG'}(u) = \out_{\bG'}(v)$. In this case, the task is denoted by
  $(\gfam,f)$ where $f:\gfam\longrightarrow L$ is the final
  labelling function.
\end{definition}

\begin{definition}
  A uniform \distask $(\gfam,f)$ is {\em quasi-covering-lifting 
    closed} if there exists a recursive function 
  $r:\covgfam\longrightarrow\N$ such that, if there exist graphs \bK,
  \bK' in \gfam and \bH such that \bK and \bK' are
    quasi-coverings of \bH of radius $r(\bK)$, then $f(\bK) = f(\bK').$
\end{definition}

\begin{theorem}[\cite{GMstructToCS}]\label{caracLTD}
  A uniform \distask is locally computable with local
  termination detection if and only if it is quasi-covering-lifting closed.
\end{theorem}

\subsection{Observed Termination Detection}

 \macro{\grso}{\grs^{\mbox{\textsc{o}}}}

\begin{theorem}\label{caracOTD}
  A \distask $T=(\gfam,S)$ is locally computable with observed
  termination detection if and only if 
  \begin{theoenum}
  \item $T$ is covering-lifting closed,
  \item \label{boundedquasicov}  there exists a recursive function
    $r:\covgfam\longrightarrow \N$ such 
    that for any $\bH\in\widehat{\gfam}$, there is no  strict
    quasi-covering  of \bH  of radius $r(\bH)$ in \gfam. 
  \end{theoenum}
\end{theorem}

\begin{proof}
  \CN
  This is actually a simple corollary of the
  quasi-lifting lemma. 

We prove this by contradiction. We assume there is a graph relabelling
system \grs  with observed termination detection that
computes the specification $S$ on \gfam.

Now we suppose there exists $ \bH\in\widehat{\gfam}$ that admits strict
quasi-coverings of unbounded radius in \gfam. By Lemma
\ref{covgfamnoetherian}, $\grs$ is noetherian for \bH. 
Consider an execution of $\grs$ of length $l$.

By hypothesis, there exists $\bK\in\gfam$ a strict quasi-covering of
\bH of radius $2l+1$. By the quasi-lifting lemma, we can simulate on a
ball of radius $2l+1$ of \bK the execution of
\grs on \bH. At the end of this relabelling steps,
there is a node in \bK that is labelled \term. As the quasi-covering
\bK is strict, there exists at least one node outside of the ball that
has not even taken a relabelling step of \grs, hence that has not
written anything to \out. Hence \grs has
not the observed termination property on \bK. A contradiction.

\CS
In some sense, we will observe the
termination of \grsi by letting \rtrust increase a bit more. In order to
do that, we have to relax the condition $\varphi_I$. 

We define the condition 
$\varphi_O$ by\footnote{with the convention that -- in order to avoid
  the problems of definition of \bH, or its belonging to \covgfam -- in the
{\em or} conditions, the right part is not ``evaluated'' if the left
 part is true.}:
\begin{itemize}
\item $\dxi(\bH,n(v_0),r^{\mathrm{t}}(v_0))\neq Yes$ or
  $r^{\mathrm{t}}(v_0) \leq r(\bH),$  
\item $\dxi(\bH,n(v_0),r^{\mathrm{t}}(v_0))\neq Yes$ or
  $\out(v_0)=\out_{\choosen(v_0)}\left(n(v_0)\right).$ 
\end{itemize}

In order to define \grso, we add to $\algo\left((\bH,
  \choosen),\varphi_\mathrm O\right)$ the following rule:

  \begin{rrule}{\grso}  
    \ritem[\grso-\bH']{Termination Detection and Pick an Output}{
    \item $\dxi(\bH,n(v_0), r^t(v_0)) = Yes.$
    }{
    \item $\choosen'(v_0)=\bH'.$
    }
    \ritem[\grso]{Unifying Outputs}{
    \item for all $v\in B(v_0)$, $\choosen(v_0)\preceq\choosen(v).$
    }{
    \item for all $v\in B(v_0)$, $\choosen'(v)=\choosen(v_0).$
    \item for all $v\in B(v_0)$, $\out'(v)=\out_{\choosen(v_0)}\left(n(v)\right).$
    }

    \ritem[\grso]{Termination Detection}{
    \item $r^{\mathrm{t}}(v_0) > r(\bH).$
    }{
    \item $\lterm'(v_0)=\term.$
    }

  \end{rrule}

This system  computes the task $(\gfam,S)$ and has an observed
termination detection.

First, \grso is noetherian. The valuation is now slightly different of
the one of \carto, but we can use the same proof as for Theorem
\ref{dxinoeth} to prove that $\left((\bH,
  \choosen),\varphi_\mathrm O\right)$ is uniform. Here again, the GSSP
Rule will stop being enabled on each vertex for the
same value of $r^{\mathrm{t}}$, the one that is equal to $r(\bH)+1$. 

Now, suppose we have, at a given time $i$, on a node $v$, $r(\bH(v)) <
r^t(v)$, then, by the hypothesis \ref{boundedquasicov} and by the
Remark~\ref{nonstrictqcovarecov}, the entire graph $\bG_i$ is a covering of  
\bH. Hence \mk is terminated. Furthermore, the second 
precondition of $\varphi_O$ ask $\out$ to be computed on each vertex,
from the same graph \bH' as $\choosen$ is a component of the
valuation. 

Thus the detection of termination is correct. Moreover, by
covering-closure, the \out labels are correct for the specification
$S$. 
\end{proof}

In the following we refer to hypothesis \ref{boundedquasicov} as the {\em
  relatively bounded radius of quasi-covering condition}.

\subsection{Global Termination Detection}

In this section, we characterise the most demanding termination mode.
\macro{\grsg}{\grs^{\mbox{\textsc{g}}}}
\begin{theorem} \label{caracGTD}
  A \distask $(\gfam,S)$ is locally computable with global
  termination detection if and only if
  \begin{itemize}
  \item any labelled graph in \gfam is covering minimal,
  \item   there exists a recursive function $r:\gfam\longrightarrow \N$ such
    that for any $\bG\in\gfam$, there is no quasi-covering of
    \bG  of radius $r(\bG)$ in \gfam, except \bG. 
  \end{itemize}
\end{theorem}

\begin{proof}
  \CN
  We need only to prove the first item. As
  minimality implies $\gfam = \covgfam$, the second
  one is a restatement of the one for termination detection by
  observer.  

  The minimality of any graph in \gfam is again a corollary of the
  lifting lemma. Suppose there are \bG and \bH in \gfam such that \bG
  is a strict covering of \bH. 

  We consider a relabelling chain in \bH. It comes from the lifting
  lemma that this can be lifted step by step in \bG. When the final
  step is reached in \bH, and as \bG is a strict covering of \bH,
  there are at least two nodes in \bG where to apply the final \term
  rule. Hence a contradiction.

  \CS
  The two hypothesis imply that task $(\gfam,S)$
  has the observed termination detection property (the
  cove\-ring-lifting property is a trivial tautology when all
  concerned graphs are minimal).  Hence there exists a relabelling
  system \grso that computes $S$ with OTD.

  We define
  \grsg the relabelling system obtained by the union of $\grso$
  without the ``Termination Detection'' rules and the
  rules of Section \ref{subsec:spantree} for the computation of a
  spanning tree. The root is the vertex that gets number $1$ in \mk,
  when this vertex observes the termination for \grso with the 
  following rule:

  \begin{rrule}{\grsg}
  \ritem[\grsg]{Root}{
    \item $r^{\mathrm{t}}(v_0) > r(\bH)$,
    \item $n(v_0) = 1$.
    }{
    \item $\mem'(v_0)=\epsilon.$
    }
  \end{rrule}

  By minimality of \bG, there is only one vertex with number $1$ when
  \mk is finished. Hence we really get a spanning tree and not a
  spanning forest.
\end{proof}

\section{Applications}

In this section, we present consequences from the previous theorems. There
are known computability results, some new ones and the proof that the
different notions of termination detection are not equivalent.

We emphasise that the following results are bound to the
model of local computations. Results on other models should be similar
even if strictly and combinatorially speaking different. They remain
to be precisely described and computed.

\subsection{Domains and Specifications}

Consider a locally computable task $\task=(\gfam,S)$. The first remark
is that implicit termination and LTD give conditions on the
specification (with respect to the domain) but there are (sometimes
trivial) tasks on any domain. And on the contrary, OTD 
and GTD have conditions upon the domain and (weak) conditions on the
specification. The difference between domains that have OTD for
(almost) any task and the ones that have only GTD for any tasks
depends upon the covering-minimality of the graphs in the given
domain. 

As a conclusion, with respect to the termination detection
criteria, whether we can work where we want but we cannot do what we
want (the specification has to respect covering-lifting and
quasi-covering-lifting closures), whether we can do 
whatever we want, but we can do it only on 
particular families of networks. The more interesting possible
trade-off is probably on the LTD tasks but that is the most complex
families and its complete characterisation has yet to be done.

\subsection{Known Results as Corollaries}

We first sum up some results for every category of termination
detection. Then we show that the hierarchy is strict. With the remark
from the previous subsection, we focus mainly on the relevant part
(domain or specification). A very important application for
distributed algorithms, the Election problem, is done in a dedicated
section, Section~\ref{sec:election}.

\subsubsection{Implicit Termination.}
From Th.\ref{caracITD}, we can see that what can be computed with
implicit termination depends only of what is kept whenever
there is lifting. Such a property is called ``degree-refinement'' in
the graph-theoretic context \cite{Leighton}. Hence, what can be
computed with implicit termination is exactly a computation about the
degree-refinement of the network. See \cite{GMMrecognition} about an
investigation of the decision task of recognising whether the
underlying network belongs to a given class.

\begin{example}
  We denote by $R$ the family of rings.
  Consider the following task $\task_1=(\graph,\chi_{R})$
  which asks to decide whether the network is a ring or not. The
  \distask $\task_1$
  is locally computable with implicit termination but not with a
  relabelling system with LTD. Consider chains. Long ones are
  quasi-coverings of arbitrary radius for a given ring. Hence
  $\chi_{R}$ is not quasi-covering-lifting closed.
\end{example}

We give a second example with domain $R$.

\begin{example}
  We denote \textsc{Div} the following 
  specification: the \out labels are taken in \N and $\bG
  \textsc{Div}\bG'$ if and only if the final \out label divides the
  size of \bG. 

  $(R,\textsc{Div})$ is covering-lifting closed as a ring \bG is a
  covering of a ring \bH if and only if the size of \bG divides the
  size of \bH. However, \textsc{Div} is not quasi-covering-lifting
  closed. There are ``huge'' minimal rings that are quasi-covering of
  arbitrary radius of, say, $R_7$. 
\end{example}

\subsubsection{Local Termination Detection.}
See \cite{GMstructToCS} for numerous examples about of the computation
of a structural knowledge (that is a uniform labelling) from another
one.

\begin{example}
  The relation  $\mbox{\textsc{Colo}}_3$ is the specification of the
  3-colouring problem. The task $\task_2=(R,\mbox{\textsc{Colo}}_3)$ has
  local termination detection (relabelling system given in
  Example~\ref{colod}) but has not observed termination detection
  for there are ``huge'' rings that are quasi-covering of any given arbitrary
  radius of, say, $R_3$.   
\end{example}

\subsubsection{Observed Termination Detection.}
Here we will find the frequent (sometimes implicit) assumptions usually made
by distributed algorithms:

\begin{itemize}
\item size or diameter is known,
\item a bound on the diameter or the size is known.
\end{itemize}

It shall be noted that the computability results from the work of
Yamashita and Kameda belong to this category.

\begin{example}
  Let $n\in\N, n\geq 6$. We note ${R}^n$ the rings of size at most $n$.
  We consider $\task_3=({R}^n,\mbox{\textsc{Colo}}_3)$. The
  radius of strict quasi-covering are bounded in ${R}^n$. Hence
  $\task_3$ has OTD, but it has not GTD, for the ring $R_6$ is not
  covering-minimal.  
\end{example}

\subsubsection{Global Termination Detection}
Here we really find all the well known assumptions usually made
about distributed network algorithms. 
The theorems  admit well known corollaries; more precisely from  
Theorem \ref{caracGTD} we deduce immediately that  we have global
termination detection {\em for any task} for the following families of
graphs:  
\begin{itemize}
\item graphs having a leader,
\item graphs such that each node is identified by a unique name,
\item trees.
\end{itemize}

From Theorem \ref{caracGTD} we deduce there is no 
such termination for:
\begin{itemize}
\item the family of covering-minimal anonymous rings,
\item the family of covering-minimal anonymous networks.
\end{itemize}

\begin{example}
  Let $n\in\N$. We note ${PR}^n$ the rings of prime size at
  most $n$. We consider $\task_4=(\mbox{\textsc{Colo}}_3,{
  PR}^n)$. The radius of quasi covering are bounded in ${
  PR}^n$, and rings of prime size are covering-minimal. Hence
  $\task_4$ is in $\task_{GTD}$. 
\end{example}

\subsection{The Hierarchy is Strict}

The previous examples $\task_1$, $\task_2$ and $\task_3$ show that
the hierarchy is strict and that the four notions of termination are
different. 

\begin{proposition}\label{stricthierarchy} 
  \begin{eqnarray*}
  {\tasks}_{\mathrm{ GTD}}(\graph)={\tasks}_{\mathrm{ OTD}}(\graph)=\emptyset\subsetneq{\tasks}_{\mathrm{ LTD}}(\graph)&\subsetneq &{\tasks}_{\mathrm I}(\graph),\\
  {\tasks}_{\mathrm{ GTD}}(R)={\tasks}_{\mathrm{ OTD}}(R)=\emptyset\subsetneq{\tasks}_{\mathrm{ LTD}}(R)&\subsetneq &{\tasks}_{\mathrm I}(R),\\
  {\tasks}_{\mathrm{ GTD}}(R^n)=\emptyset&\subsetneq&{\tasks}_{\mathrm {OTD}}(R^n).
  \end{eqnarray*}
\end{proposition}

\subsection{New Corollaries}
New interesting corollaries are obtained from these theorems. 

\subsubsection{Multiple leaders}

 From Theorem~\ref{caracOTD} and Lemma~\ref{techlemma}, we get
 \begin{corollary}
   Any covering-lifting closed task has an OTD solution in the following
   families:
   \begin{itemize}
   \item graphs having exactly $k$ leaders,
   \item graphs having at least one and at most $k$ leaders.
   \end{itemize}
 \end{corollary}

 From Theorem \ref{caracOTD}, we deduce a negative result for
 the family of graphs having at least $k\geq2$ leaders.

\subsubsection{Link Labellings and Sense of Direction.}
We recall that a homomorphism $\varphi$ 
from the labelled graph $\mathbf G$ to 
the labelled graph $\mathbf G'$
is a graph homomorphism from $G$ to
 $G'$ which preserves the labelling: a node is mapped
to a node with the same label and a link is mapped to a link
with the same label. 

Thus,  a family
of labelled graphs induced by a weak sense of direction  satisfies
the condition \ref{boundedquasicov} of Theorem \ref{caracOTD} (indeed
weak sense of direction forbids quasi-coverings). Thus, for any task,
observed termination detection is possible in all families of graphs
with weak sense of direction.

\subsection{About the Complexity of Local Computations}
The step complexity of \mk is $O(n^3)$ \cite{SelfstabEnum}. Denote $C$ the
complexity of GSSP in the bounded radius of quasi-covering
context. Hence we can see that the complexity of a \distask
is bounded by $O(n^2+C)$. It is easy to see that the complexity of
GSSP is closely related to the bound $r$ of the radius of
quasi-coverings. When \mk is terminated, any node has to go from $0$
to $r$ with GSSP rule. Thus $C\leq n\times (r+1)$. 

Whether the complexity comes from the distributed gathering of
information or from the termination detection depends upon the order
of magnitude of $r$.

A similar study of the complexity of distributed algorithms by
upper-bounding by ``universal algorithm'' is done in \cite{BVholo}
where, it shall be noted, the notion of quasi-covering is introduced
for trees.

\section{A Characterisation of Families of Networks
in which Election is Possible}
\label{sec:election}
\macro{\election}{{\small\textsc{Election}}}
\macro{\elect}{{\small\textsc{Elect}}}
\macro{\nonelect}{{\small\textsc{Non-Elect}}}

Considering a labelled  graph, we say informally that a given vertex
$v$ has been elected when the graph is in a global state such that 
exactly one  vertex has the label
\elect  and all other vertices have the 
label \textit{\nonelect.} The labels $\it {\elect}$ and
\textit{\nonelect} are terminal, i.e., when they appear on a vertex
they remain until the end of the computation. This is the standard
definition. 

Note that if we  ask nothing about the non elected vertices, this gives
an equivalent definition in terms of computability. Because when a node
is elected, it can broadcast it to all the nodes of the networks.

\begin{definition}
  Let $\gfam$ be a class of connected  labelled graphs.
Let $\cal R$ be a locally generated relabelling relation, we say 
that $\cal R$ is an election algorithm for the class $\gfam$ if $\cal R$
is noetherian and for any graph $\mathbf G$ of $\gfam$ and
for any normal form $\mathbf G'$ obtained from $\mathbf G,$ 
$\mathbf G {\cal R}^* \mathbf G',$ there exists exactly one vertex
with the label ${\elect}$ and all other vertices have the label 
{\nonelect.}
\end{definition}

With the notation of the previous part, we have the various
definitions for the various kinds of termination detection.

\begin{definition}
  Let $\gfam$ be a class of connected  labelled graphs. Let $\election$ be the
  following relation: \bG and \bG' are in relation if and only if
  there exists in \bG' exactly one vertex  with the label ${\elect}$
  and all other vertices have the label  
  {\nonelect.} 

  The implicit(resp. LTD, OTD, GTD )-\election on \gfam is
  the task  $(\gfam,\election)$ with implicit (resp. local, observed, global)
  termination detection. 
\end{definition}

We underline that we are looking for classes of networks that admit
the same \election algorithm for all its elements. Having an algorithm
that works for several networks (say, independently of the knowledge of
its size) is very important for reliability. In this setting, saying
that \bG admits an \election algorithm amounts to say that
$(\left\{\bG\right\},\election)$ is a computable task. It is important
to note that saying that 
\election is computable on a given family \gfam {\em does not} mean
that $(\left\{\bG\right\},\election)$ is a computable task for any
$\bG\in\gfam$, but means that $(\gfam,\election)$ is a computable
task. 

\medskip
We can see that the definition of LTD-\election is equivalent
to the standard definition of \election. 

We will prove that the possibility of the LTD-\election on \gfam is
equivalent to the possibility of the GTD-\election. But
first we give two examples of elections.

\subsection{Two Examples}\label{exemple2}
\subsubsection{An Election Algorithm in the Family of Anonymous Trees.}

 The following relabelling system elects in trees.
  The set of labels is $L=\{N,\elect,\nonelect\}$.  
  The initial label on all vertices is   $N$. 

  \begin{rrule}{\textsc{Elec\-tion\_Tree}}
    \ritem{Pruning rule\label{prune}}{
    \item $\lambda(v_0)=N$,
    \item $\exists!\; v\in B(v_0,1), v\neq v_0, \lambda(v)=N$.
      }{
    \item $\lambda'(v_0)=\nonelect.$
      }
    \ritem{Election rule\label{elec}}{
    \item $\lambda(v_0)=N$,
    \item $\forall v\in B(v_0,1), v\neq v_0, \lambda(v)\neq N$.
      }{
    \item $\lambda'(v_0)=\elect.$
      }
  \end{rrule}

  Let us call a pendant vertex any vertex labelled $N$ having
  exactly one neighbour with the label $N.$
  There are two meta-rules ${\sc Election\_Tree1}$ and 
  ${\sc Election\_Tree2}.$
  The  meta-rule ${\sc Election\_Tree1}$
   consists in cutting a pendant vertex by giving
  it the label 
  $\nonelect.$ The  label $N$ of a vertex $v$ 
  becomes $\elect$ by the meta-rule ${\sc Elec\-tion\_Tree2}$ 
  if the vertex $v$ has no
  neighbour labelled $N.$ A complete proof of this system
  may be found in \cite{LMS1}.
\subsubsection{An Election Algorithm in the Family of Complete Graphs.}

 The following relabelling system elects in complete graphs.
  The set of labels is $L=\{N,\elect,\nonelect\}$.  
  The initial label on all vertices is   $l_0=N$. 

  \begin{rrule}{\textsc{Election\_Complete-graph}}
    \ritem{Erasing rule}{
    \item $\mem(v_0)=N$,
    \item $\exists\; v\in B(v_0,1), v\neq v_0, \mem(v)=N$.
      }{
    \item $\mem'(v_0)=\nonelect.$
      }
    \ritem{Election rule}{
    \item $\mem(v_0)=N$,
    \item $\forall v\in B(v_0,1), v\neq v_0, \mem(v)\neq N$.
      }{
    \item $\mem'(v_0)=\elect.$
      }
  \end{rrule}

It is straightforward to verify that this system elects in the family of
complete graphs.

\subsection{Characterisation of Election}

We  show that the LTD-\election is
solvable if and only if the GTD-\election is solvable. Then we use the
general characterisation of this paper to conclude.

\begin{proposition}
  Let \gfam be a labelled graph family. The LTD-\election task on
  \gfam is computable if and only if the GTD-\election is.
\end{proposition}

\begin{proof}
  The sufficient condition is easy (Proposition \ref{hierarchy}). We
  focus on the necessary condition.

  Suppose \grs is a graph relabelling relation with LTD solving the
  Election task on \gfam. In order to convert it in a  graph
  relabelling relation with GTD, we will add some rules to \grs.  We
  add a rule that starts the computation, with GTD, of a spanning tree
  rooted in the \elect   vertex. This standard construction is given in 
  Section \ref{subsec:spantree}. 
\end{proof}

\begin{remark}
  This demonstration shows that even if we define a task with a LTD
  flavour, it can reveal to be in the GTD family of tasks because of
  the form of the specification. Furthermore, we will now not
  distinguished between LTD(resp. OTD, GTD)-\election.
\end{remark}

As a corollary of Theorem \ref{caracGTD}, we get:

\begin{theorem}\label{election}
  Let $\gfam$ be a class of connected labelled graphs.
  There exists an \election algorithm for $\gfam$ if and only
  if 
  \begin{itemize}
  \item graphs of $\gfam$ are minimal for the covering relation, and
  \item there exists a computable function $r:\gfam\ra \N$ such
  that for all graph $\mathbf G$ of  $\gfam$, there is no
  quasi-covering of $\mathbf G$ of radius greater than $r(\mathbf G)$ in
  $\gfam$, except $\mathbf G$ itself. 
  \end{itemize}
\end{theorem}

\begin{remark}
  In fact, the \election algorithm can be directly derived from the
  \carto algorithm. When a node detects the termination
  of \mk, it sets its $\out$ label to \elect or \nonelect whether it
  is numbered 1 or not. 
\end{remark}

\subsection{Applications}
\label{BVelection}
The first attempt of a complete characterisation of election was first
done in \cite{BVelection}, but the results were only given when a
bound upon the diameter is initially known. In the general
no knowledge case, they give a ``pseudo''-election algorithm, \ie,
some \elect labels can appears during the computation, this is only
when the computation is finished that this label has to be
unique. This is exactly the definition of implicit-\election.

Known results appear now as simple corollaries of Theorem~\ref{election}.
\begin{itemize}
\item \cite{MazurEnum} Covering minimal networks where the size is
  known;%
\item Trees, complete graphs, grids, networks with identities.
\end{itemize}

Those last  families contains no $q$-sheeted quasi-covering of a
given graph for   $q\geq2$,
hence the $r$ function can be  twice the size of the graph, see
Lemma~\ref{techlemma}.

We also get some new results. 
An interesting result is that there is no election algorithm for
the family of all the networks where the election is possible.
\begin{proposition}
  Let \bG be a labelled graph. \election is computable on \bG
  if and only if \bG is covering-minimal.
\end{proposition}

\begin{proposition}
  There is no \election algorithm on the family of covering-minimal graphs.
\end{proposition}
\begin{proof}
  Rings with a prime size are minimal and does not respect the
  relatively bounded quasi-covering condition.
\end{proof}

However, from Theorem~\ref{caracITD}, it is easy to derive where
implicit-\election is computable.

\begin{proposition}
  \election is computable with implicit termination on the family of
  covering-minimal graphs. 
\end{proposition}

We obtained as a direct corollary:
\begin{proposition}
  There exists an election algorithm for covering minimal graphs where
  a bound of the size is known. 
\end{proposition}

We can notice that no trivial extension of the proof of the
Mazur\-kiewicz algorithm enables to obtain directly  this proposition.

\medskip
We also have a new and interesting result for graphs with at most $k$
distinguished vertices: 

\begin{proposition}\label{kdist} 
  Let $k\in\N$. Let $\mathcal I$ be a family of covering-minimal
  $\{0,1\}$-labelled graphs such that for all graph, there are at most
  $k$ vertices labelled with $1$. Then, there exists an election
  algorithm for this family.
\end{proposition}
\begin{proof}
  We define $r(\mathbf G) = (k+1)|V(\mathbf G)|$ and we remark that
  quasi-covering in $\mathcal I$ can be at most $k$-sheeted. Hence, by
  Lemma~\ref{techlemma}, we deduce that $r$ has the desired property.
\end{proof}
From this proposition we deduce that to have an \election algorithm
in a network where uniqueness of an identity is not guaranteed, we
only need a bound on the multiplicity of identities.

\section{Conclusion}

\subsection{Characterisations of termination detection}

Distributed algorithms are very different from sequential ones. How to
make them terminate is a difficult problem. Moreover in this paper, we
show that even if the termination is given, and so can be detected by
an omniscient observer, the detection of this fact is not always
possible for the nodes inside the network.

In this paper, we present a quite comprehensive description of the
computability of \distasks with explicit detection of the
termination. We show one can define four 
flavours of termination detection: implicit termination detection,
local termination detection, termination detection by a distributed
observer and global
termination detection. For each termination detection, we give the
characterisations of distributed 
tasks that admit such a termination detection, and we show they form
a strict hierarchy. The local termination detection flavour is only
characterised in the  case of uniform tasks. It has yet to be
completely investigated. 

We prove that if we ask for implicit or local termination detection,
we can work in any family of networks, but the computable tasks are
restricted. On the other hand, we show that if we ask for global
termination detection, we have to work on special classes of graphs  -
minimal graphs with relatively bounded radius of quasi-coverings - but
there, every task is 
computable. This characterisation precisely explains numerous kind of
hypothesis that are 
traditionally made when designing distributed algorithms.

In conclusion, we show that a distributed task is not only described by
a specification - a relation between inputs and outputs -, a domain - 
the family of networks in which we have to meet the specification -,
but also by the kind of termination detection we ask for.%

\subsection{Comparison with other models}

In contrast with previous works about the computability of distributed
tasks, we can say that, usually, the
termination of the distributed algorithms is ``factored out'': the nodes
know at the beginning an upper bound on the number of 
steps it will take. For Yamashita and Kameda models and Boldi and
Vigna models, it is the  
particular initial knowledge that enable to determine how many steps of
union of local views is sufficient.

It can be observed that, actually, the universal algorithms in these works
are constituted by a potentially infinite loop (merge local views for
Yamashita/Kameda and Boldi/Vigna, snapshot read-write for
Herlihy/Shavit and Borowsky/Gafni \cite{BGatomicsnapshot}) and an
external condition to say when to end the infinite loop. This
condition does not depend on the distributed computations. In this
sense, we can say that the termination is factored out: it is not
detected in a truly distributed way as the number of rounds is known
in advance, it does not depend of what is gathered by each node in the
exchange of information of the distributed algorithm.  

In a kind of contrast, we can see that our asynchronous snapshot
algorithm is constituted of two parts: Mazurkiewicz' algorithm, that
is {\em always} terminating (implicit termination); and the
generalised SSP stability detection that does not terminate alone.
That is this combination that enables to detect, in a truly
distributed way, the termination of the distributed tasks. When to
stop GSSP is computed from the value obtained by Mazurkiewicz'
algorithm, and not from a given a priori value like in the other
approachs.

\subsection{Impossibility results in non-faulty networks}

The results given in this paper show that there are also
possibility/impossibility results even with non-faulty networks. This
paradox could be settled in the recent approach of failure detectors:
the various kind of distributed systems can be seen as a perfect
system (synchronous, centralised, with identities ...) with various failure
(asynchronicity, node failures, communication failures,
...)\cite{ChandraToueg,GafniRoundByRound}. In this 
contribution, we show that lack of structural knowledge (nodes do not
know exactly what is the topology of their network), and lack of
structural information (e.g. unique identities) are also a kind of
failure in this concern.

\bigskip
The authors wish to thank the anonymous referees for some helpful
comments. They are also specially grateful to Bruno Courcelle, Pierre
Cast\'eran and Vincent Filou for their corrections and stimulating questions
regarding the previous version of this report.

\bibliographystyle{alpha}
\bibliography{terminaison}
\end{document}